\documentclass[preprint,11pt]{elsarticle}

\journal{Journal of Computational Physics}

\usepackage{amsmath, amsthm, amssymb}
\usepackage{color}
\usepackage{listings}
\usepackage{a4wide}
\usepackage{courier}
\usepackage{comment}
\usepackage{amsfonts}
\usepackage{graphicx}
\usepackage[utf8]{inputenc}
\usepackage[hang,small]{caption}

\theoremstyle{definition}
\newtheorem{mydef}{Definition}
\theoremstyle{plain}
\newtheorem{mylem}{Theorem}
\begin{document}
\begin{frontmatter}

\title{Parallel family trees for transfer matrices in the Potts model}


\author[uchile,cecs]{Cristobal A. Navarro\corref{cor1}}
\ead{crinavar@dcc.uchile.cl}
\author[cecs]{Fabrizio Canfora}
\author[uchile]{Nancy Hitschfeld}
\author[uchile]{Gonzalo Navarro}

\address[uchile]{Department of Computer Science, Universidad de Chile, Santiago, Chile.}
\address[cecs]{Centro de Estudios Científicos (CECs), Valdivia, Chile.}

\cortext[cor1]{Corresponding author}

\begin{abstract}
The computational cost of transfer matrix methods for the Potts model is related to the question 
\textit{into how many ways can two layers of a lattice be connected?}. 
Answering the question leads to the generation of a combinatorial set of lattice configurations. This set defines the \textit{configuration space} of the problem, 
and the smaller it is, the faster the transfer matrix can be computed.
The configuration space of generic $(q,v)$ transfer matrix 
methods for strips is in the order of the Catalan numbers, which grows asymptotically as $O(4^m)$ 
where $m$ is the width of the strip.
Other transfer matrix methods with a smaller configuration space indeed exist but they make assumptions 
on the temperature, number of spin states, or restrict the structure of the lattice. 
In this paper we propose a parallel algorithm that uses a sub-Catalan configuration space of $O(3^m)$
to build the generic $(q,v)$ transfer matrix in a compressed form.
The improvement is achieved by grouping the original set of Catalan configurations into a forest 
of family trees, in such a way that the solution to the problem is now computed by 
solving the root node of each family. As a result, the algorithm becomes exponentially faster 
than the Catalan approach while still highly parallel. 
The resulting matrix is stored in a compressed form using $O(3^m \times 4^m)$ of space, making 
numerical evaluation and decompression to be faster than evaluating the matrix in its 
$O(4^m \times 4^m)$ uncompressed form.
Experimental results for different sizes of strip lattices show that the \textit{parallel family trees (PFT)} strategy 
indeed runs exponentially faster than the \textit{Catalan Parallel Method (CPM)}, 
especially when dealing with dense transfer matrices. 
In terms of parallel performance, we report strong-scaling speedups of up to $5.7X$ when running on an 
8-core shared memory machine and $28X$ for a 32-core cluster.
The best balance of speedup and efficiency for the multi-core machine was achieved when using $p=4$ processors, 
while for the cluster scenario it was in the range $p \in [8,10]$. 
Because of the parallel capabilities of the algorithm, a large-scale execution of the parallel family 
trees strategy in a supercomputer could contribute to the study of wider strip lattices. 
\end{abstract}

\end{frontmatter}

\section{Introduction}
The Potts model \cite{potts} has been widely used to study physical phenomena of 
\textit{spin lattices} such as phase transitions \cite{PhysRevLett.43.799} in the 
thermodynamical equilibrium. Lattices such as square, triangular, honeycomb and kagome 
are of high interest and are being studied frequently 
\cite{Chang_Shrock_2000, Shrock_Tsai_1999, Chang_Salas_Shrock_2002, Chang_Jacobsen_Salas_Shrock_2002}. 
When the number of possible spin states is set to $q=2$, the Potts model becomes the classic 
Ising model \cite{ising_1925}, which was solved by Onsager \cite{NYAS:NYAS627} for 
the infinite-volume limit on a torus. 
For higher values of $q$ the problem becomes much harder and no solution has been found yet.
Nevertheless, it is of interest to study the problem in the form of a strip lattice. 
Hopefully, the study of sufficiently wide strips could contribute at understanding 
the physical properties of such complex systems under different boundary conditions.

An effective technique for obtaining the partition function of \textit{strip lattices} 
is to compute its transfer matrix, denoted $M$. The transfer matrix technique allows the study of 
strips that repeat their lattice structure along one of its dimensions. 
$M$ can be computed symbolically or numerically (fully or partial) 
evaluated on $(q,v)$. When there is enough disk space, we find that it is more convenient 
to compute $M$ using polynomials on $(q,v)$. 
Indeed, computing $M$ with general $(q,v)$ has an impact on performance and memory, 
but it gives the advantage that $M$ will not have to be re-computed many times when doing 
numerical sweeps for $q$ and $v$. Another advantage is that from the general $(q,v)$ transfer 
matrix one can generate many partially evaluated instances of the transfer matrix that can be used later for numerical sweeps on the remaining parameter. 
For limited computational resources, generating $M$ partially or fully evaluated is a practical choice. 

If the strip lattice represents an infinite band, then analysis can be performed by computing 
the eigenvalues of $M$.
If the strip lattice is finite, then a initial condition vector $\vec{Z_1}$ is needed. In that 
case, boundary conditions have to be specified. Typical boundary conditions are free, periodic, cylindrical and cyclic.
$M$ and $\vec{Z_1}$ together form a partition function vector $\vec{Z}$ based on the following recursion:
\begin{equation}
	\vec{Z}(n)=M\vec{Z}(n-1) = \vec{Z} = M^{n-1}\vec{Z_1}
	\label{eq_Z}
\end{equation}
Computing the powers of $M^{n-1}$ is done in a numerical context, otherwise memory usage would become intractable. 
When $M^{n-1}$ is computed, the first element of $\vec{Z}$ becomes the partition function of the strip lattice.

This work focuses on the process of building $M$, which is an \textit{NP-hard} 
problem \cite{Woeginger:2003:EAN:885909.885927} where exponential cost 
algorithms are involved in the process, with the width $m$ as the exponent. 
There are different approaches for building $M$:
(1) In the \textit{spin representation} approach, an integer value is chosen for $q$ 
and the transfer matrix $T$ is obtained by combining the different spin configurations in the graph 
layer. Under this approach, the size of $M$ becomes $q^{|V|} \times q^{|V|}$, where $|V|$ is the number 
of spins in the layer of the strip. A more detailed explanation on the spin representation 
approach is available in the first of the six works by Salas, Sokal and Jacobsen series of papers \cite{salas_sokal_1}.
(2) One can also obtain $M$ as a product of sparse matrices of asymptotic size $O(4^m)$ \cite{1751-8121-43-31-315002}, 
one per edge and practically linear in the number of edges, where $M$ is not constructed explicitly 
but only its action on a given vector of states.
(3) Alternatively one can compute $M$ with a generic $(q,v)$ method where 
the configuration space grows proportional to the Catalan numbers \cite{chang2009structure} 
or asymptotically as $O(4^m)$, leading to a matrix of size $O(4^m \times 4^m)$. 
Indeed there are other strategies that can achieve smaller transfer matrices \cite{salas_sokal_2011, 
MGhaemi, bedini}, but they assume special properties for the lattice, work only for finite 
graphs or need to fix the values of $v$ and/or $q$ in order to take any advantage. 
We believe it is worth studying what are the possibilities for algorithmic improvements 
in the generic $(q,v)$ Catalan based approach since it is a general method applicable to any planar strip. 

In the light of these aspects just mentioned, we ask \textit{\textbf{question 1:} 
Is there a generic $(q,v)$ method 
that can compute the transfer matrix for any planar strip lattice, using a sub-Catalan 
configuration space?}.
From our research we have found that: \textit{a hierarchical symmetry exists among elements of the configuration space that define the transfer matrix}. 
This symmetry is revealed when first applying deletion-contraction to certain edges of the strip layer. If this symmetry is used so that the 
configuration space is re-organized as a forest of hierarchical families, then a parallel computation only on the root nodes is sufficient for generating a compressed transfer matrix. 
When exploiting this symmetry, the configuration space is reduced from $O(4^m)$ to $O(3^m)$, which is an improvement to the actual bound on general transfer matrix methods 
for strips. This result allows us to answer positively to \textit{question 1}.

With the evolution of computer architectures towards a higher amount of cores \cite{blake2009survey, duncan1990survey}, parallel computing is not anymore limited to clusters or super-computing; 
workstations can also provide high performance for solving physical problems \cite{navhitmat2014}. 
It is in this last category where most of the scientific community lies, therefore parallel implementations for multi-core machines 
are the ones to have the largest impact on the community. Considering how technology is changing, we ask \textit{\textbf{question 2:} Can transfer matrix methods work in parallel for modern multi-core 
architectures and scale their performance efficiently as more processors are used?}.
Given the amount of data-parallelism on the number of root nodes, the performance of the algorithm scales efficiently as more processors are used. 
Results on a multi-core 8-core machine show a speedup of $5.7X$ is achieved when using $p=8$ 
processors, and an efficiency of $95\%$ is achieved when using $p = 4$. Results on a 32-core cluster confirm that the implementation can scale in a distributed scenario, 
achieving a speedup of $28X$ when using $p=32$ processors and an efficiency of over $90\%$ for the full range $p \in [1,32]$ when dealing with large square strips. 
We can also confirm that a compressed transfer matrix not only saves data space in comparison to the original one, but it is also faster to load 
considering that it must be first evaluated for any practical usage. In the case of cluster performance, a dynamic scheduler is mandatory in order to bypass 
potential \textit{performance valleys} that are caused by the combination of unbalanced work and a static scheduler. 
Again, this result allows a positive answer for \textit{question 2}.

The paper is organized as follows: Section \ref{seq_preliminaries_related_work} covers preliminary concepts of the Potts model, 
Section \ref{sec_relatedwork} describes related work. Sections \ref{seq_algorithm_overview} and \ref{seq_algorithm_optimizations} explain the 
algorithm and the additional optimizations. Section \ref{sec_implementation} provides details about the implementation while 
in section \ref{sec_performance} we present detailed results for running time, speedup, efficiency and knee, using different amount of processors. We also compare performance 
against the \textit{Catalan Parallel Method (CPM)} \cite{DBLP:conf/hpcc/NavarroHC13}. Section \ref{sec_validation} is devoted to the 
validation of the algorithm by computing some physical results; from limiting curves to energy and specific heat, and comparing them to the results 
obtained by other authors. Section \ref{seq_conclusions} discusses our main results and concludes the impact of our work.

\section{Preliminaries}
\label{seq_preliminaries_related_work}
Let $G=(V,E)$ be a lattice with $|V|$ vertices, $|E|$ edges and $s_i$ be the state of a \textit{spin} of $G$ with 
$s_i \in [1..q]$ and $i \in [1,|V|]$. The partition function $Z(G, q, \beta)$ is defined as
\begin{equation}
Z(G, q, \beta) = \sum_{r}e^{-{\beta}h(G_{r})}
\label{eq_potts}
\end{equation}
where $\beta = \frac{1}{K_B T}$, $K_B$ is the Boltzmann constant, $T$ the temperature and $h(G_r)$ is the energy of the lattice at a given state 
$G_r$\footnote{A state $G_r$ is a distribution of spin values on the lattice. It can be seen the a graph $G$ with a specific 
combination of spin values on the vertices.}.
The Potts model \cite{potts} defines the energy of a state $G_r$ with the following Hamiltonian:
\begin{equation}
\label{eq_hamiltonian}
h(G_r) = -J\sum_{\langle i,j\rangle \in G_r}\delta_{s_{i}, s_{j}}
\end{equation}
Where $\langle i,j\rangle$ corresponds to the nearest neighbor edge from vertex $v_i$ to $v_j$, $r \in [1..q^{|V|}]$, $J$ is the interaction energy ($J<0$ for 
\textit{anti-ferromagnetic} and $J>0$ for \textit{ferromagnetic}) and $\delta_{s_{i}, s_{j}}$ corresponds to the \textit{Kronecker delta} evaluated at the pair of spins $\langle i, j \rangle$ with states $s_i,s_j$ and expressed as
\begin{equation}
\delta_{s_{i},s_{j}} = 
\left\{ 
\begin{array}{rl}
 1 &\mbox{ if $s_{i}=s_{j}$} \\
 0 &\mbox{ if $s_{i}\not=s_{j}$}
\end{array} 
\right.
\end{equation}
As the lattice becomes larger in the number of vertices and edges, the computation of equation (\ref{eq_potts})
becomes rapidly intractable with an exponential cost of $\Theta(q^{|V|})$. In practice, one can use equivalent methods that, while still exponential, in practice run 
faster than the original definition.

The \textit{deletion-contraction} method \cite{Wilf:2002:AC:560438}, or DC method, was initially used to compute the Tutte polynomial \cite{tutte} and 
was then extended to the Potts model after a relation of duality was found between the two (see \cite{potts_tutte_relation, sokal_2005}). 
DC re-defines $Z(..)$ as the following recursive equation:
\begin{equation}
Z(G, q, v) = Z(G-e, q, v) + vZ(G/e, q, v)
\label{eq_deletion_contraction}
\end{equation}
Where $G-e$ is the \textit{deletion} operation, $G/e$ is the \textit{contraction} operation and 
the auxiliary variable $v = e^{-\beta J} - 1$ makes $Z(..)$ a polynomial.
There are three special cases where DC can perform a recursive step with linear cost:
\begin{equation}
Z(G, q, v)=
\left\{ 
\begin{array}{ll}
(q+v)Z(G/e, q, v) ;		&\mbox{if \{e\} is a spike.}\\
(1+v)Z(G-e, q, v) ;		&\mbox{if \{e\} is a loop.}\\
q^{|V|};			&\mbox{if $E=\{\emptyset\}$.}
\end{array} 
\right.
\label{eq_deletion_contraction_optimizations}
\end{equation}
The computational complexity of DC has a direct upper bound of $O(2^{|E|})$.
When $|E| >> |V|$ a tighter bound is known based on the 
Fibonacci sequence complexity \cite{Wilf:2002:AC:560438}; $O((\frac{1+\sqrt{5}}{2})^{|V|+|E|})$. 
In general, the time complexity of DC can be written as
\begin{equation}
T(G) = min\Bigg( O(2^{|E|}), O\Big(\frac{1+\sqrt{5}}{2}\Big)^{|V|+|E|}\Bigg)
\end{equation}

A \textit{strip lattice} is a bidimensional graph $G=(V,E)$ that repeats its pattern at least along one dimension. It can be built as the concatenation of layers $K_1, K_2, ..., K_n$ sharing their boundary vertices and edges. Figure \ref{fig_strip_lattice_general} illustrates how the notion of 
strip lattice applies to the case of the square and kagome lattices.  
\begin{figure}[ht!]
\includegraphics[scale=0.625]{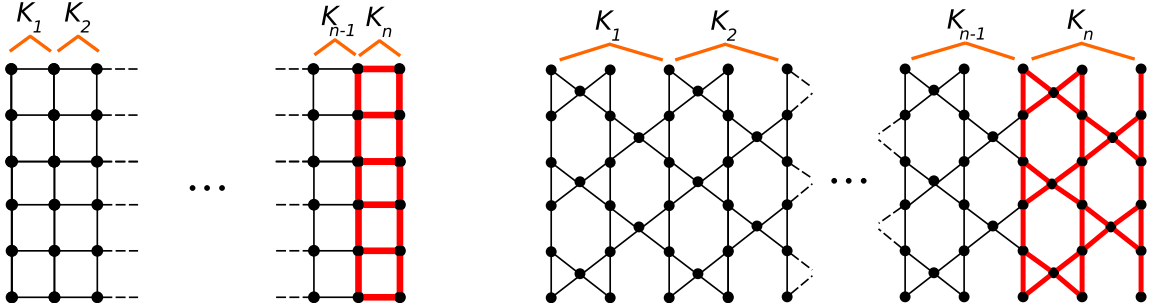}
\centering
\caption{The strip structure for the square and kagome lattices, both with a width (vertical) of $m=6$.}
\label{fig_strip_lattice_general}
\end{figure}
The transfer matrix, denoted $M$, takes advantage of the repeating nature of the lattice, allowing the study of very long graphs.
In the limit of infinite length the free energy per site becomes:
\begin{equation}
f = \frac{1}{n_{K}} ln \lambda_{+}
\label{eq_freeenergy_density}
\end{equation}
where $n_{K}$ is the number of non-shared vertices per layer and $\lambda_+$ is the dominant eigenvalue 
of $M$ with nontrivial coefficient associated.
The dimension of $M$ grows proportional to a combinatorial function $\Gamma(m)$, which depends on the size of the base (\textit{i.e.}, the width of $G(V,E)$) 
and it represents \textit{the different ways in which two layers can connect} by combining spin states and identifications. The set of configurations generated by the base corresponds to the \textit{configuration space} of the problem. 
The computational cost of a transfer matrix method comes from two sources; (1) the size of the configuration space and (2) the cost of the local algorithm. 
The sequence generated by $\Gamma(m)$ corresponds to the size of the configuration space of the problem and, as mentioned 
earlier, it defines the size of $M$. The local algorithm is in charge of computing the partition functions for each element of the configuration space.

\section{Related Works}
\label{sec_relatedwork}
The transfer matrix methods were introduced by Derrida \textit{et. al.} in 1980 \cite{derrida1980transfer} as an approach to study percolation and phenomenological re-normalization. In 
1982, Baxter used transfer matrix techniques in his seminal works as a tool for solving statistical mechanics problems \cite{baxter1982exactly}.
Salas, Sokal and Jacobsen have greatly contributed with a series of results, plus 
an additional unnumbered one that follows the same line, in which they study the physics of square 
and triangular strip lattices through the transfer matrix technique 
\cite{salas_sokal_1, Jacobsen2001701, jacobsen2003, jacobsen2006, 1179.82040, salas_sokal_2011, 
jacobsen2007}. In those works, the authors use different types of algorithmic optimizations for 
the construction of $M$ based on the symmetries available. Different scenarios are considered 
along the works, such as the zero temperature (chromatic polynomial) case, ferromagnetic and 
antiferromagnetic cases, and different boundary conditions such as free, periodic, cylindrical 
and a special boundary condition that consists of adding two extra vertices on the sides of the 
strip. Some of the contributions made in these works include the use of non-nearest neighbors 
partitions for $v=-1$, sparse matrix factorization, algebraic input from the
representation of the Temperley-Lieb algebra, symmetries for different boundary conditions 
and the computation of the limiting curves or partition function zeroes for the different 
boundary conditions up to $m \le 13$.
State of the art works on the square lattice normally study strips in the range $3 \le m \le 13$.  
For the case of the square lattice with free boundary conditions, Salas \textit{et. al.} 
achieved $m = 12$ using $v=-1$ \cite{1179.82040}. It should be noted 
that if $v \not= -1$ and free boundary conditions are used, then the configuration space is 
the one proportional to the Catalan numbers and the problem becomes computationally harder 
to handle.
The problem of the matrix size has also been improved by algebraic techniques \cite{MGhaemi} in 
the spin representation, reducing the matrix size when working with $q=2$ and $q=3$. The authors studied 
the square and triangular strips with layers of up to $r=11$ spins, which is equivalent to a square strip 
of width $m\approx5$.
Jacobsen \textit{et. al.} have studied the $q$-state Potts model for $q = 4 cos^2(\pi/p)$ being a 
Beraha number with $p > 2$ and integer \cite{jacobsen2006}. In the work, the authors study strips 
of widths in the range $m \in [2, 6]$. The relevance of their work is that they manage to compute 
the partition function using the RSOS representation.
\'Alvarez \textit{et. al.} \cite{Alvarez_Canfora_Reyes_Riquelme_2012} have reported exact 
results for the kagome strip of width $m=5$ using the generic $(q,v)$ Catalan based transfer matrix technique.
In contrast to these related works, we are interested in exploring a general $(q,v)$ method that can allow the study of strips in the state of the art range for free boundary conditions using generic $(q,v)$. For simplicity, we will restrict our physical results just to the computation and validation of the limiting curves using free boundary conditions in order to stay within the scope of our work, but not restrict the proposed strategy to these conditions. 


More general methods for computing the exact partition function of a lattice have also been 
proposed \cite{hartmann-2005-94, bedini, Shrock_2000}. 
Bedini \textit{et. al.} \cite{bedini} proposed a transfer matrix method for computing 
the partition function of arbitrary graphs using a tree-decomposed transfer matrix technique. 
For arbitrary graphs, they mean any type of finite graph; \textit{i.e.}, random or regular 
planar/non-planar graphs.
In their work, the authors obtain a sub-exponential complexity when processing random planar 
graphs. Their algorithm is considered the best so far for arbitrary graphs and the authors manage 
to achieve results for regular lattices of up to $18 \times 18$ sites. 
If the tree-decomposed transfer matrix method is applied to a strip, the configuration space 
to explore becomes the same as the traditional transfer matrix methods for strips, \textit{i.e.}, 
the tree-width becomes the width of the strip and the cost is proportional to the Catalan 
number of the tree-width. The work is closely related to another result by Jacobsen in which 
large regular lattices of up to $20 \times 21$ sites
were studied \cite{1751-8121-43-31-315002} by using a sparse transfer matrix method based on the 
product of sparse matrices, of dimension $~3^m$ for $v=-1$ and $~4^m$ for $v\not=-1$. 
The work of Haggard \textit{et. al.} \cite{haggard_computing_tutte_polynomials} is considered 
to have the best implementation of a deletion-contraction technique for the computation of the 
Tutte polynomial for any arbitrary graph (the Tutte polynomial is the dual of the partition function \cite{potts_tutte_relation}). Their algorithm reduces the computation tree in the presence of loops, multi-edges, cycles and biconnected graphs (as one-step reductions). By using a cache, 
some computations can be reused (\textit{i.e.}, sub-graphs that are isomorphic to the ones stored in the cache do not need to be computed again). 
An alternative algorithm to Haggard \textit{et. al.} was proposed by Bj{\"o}rklund \textit{et. al.} \cite{DBLP:conf/focs/BjorklundHKK08} which achieves exponential time only in the number of vertices;
$O(2^nn^{O(1)})$ with $n=|V|$. Asymptotically their method is better than deletion-contraction
considering that many interesting lattices have more edges than vertices. 
However, Haggard \textit{et. al.} \cite{haggard_computing_tutte_polynomials} have stated 
that the memory usage of Bj{\"o}rklund's method is too high for practical use.
These techniques, which are more general than the ones from the beginning of this section, cannot be directly compared against the classic transfer matrix approach, nevertheless they still needed to be mentioned as part of the related work background. General techniques compute the transfer matrix efficiently for arbitrary graphs, but do not take advantage of the regular graph structure when it is available. On the other hand, classic transfer matrix methods for strips indeed take advantage of the regular graph structure but for arbitrary graphs are not so efficient because for each layer there is a new non-sparse transfer matrix to be computed. Both strategies play an important role in the study of spin lattices. In our case, we focus on strips with regular graph structure, therefore our approach should be considered as a classic transfer matrix method.

Research on transfer matrices for strip lattices 
in the Potts model have not reported experimental results on the parallel performance, 
except for a prior work of the authors \cite{DBLP:conf/hpcc/NavarroHC13} that consists of a parallel method 
for computing general $(q,v)$ transfer matrices using the Catalan approach, which will be named 
the \textit{Catalan Parallel Method (CPM)} for the ease of referencing it later on. 
The CPM method was successfully used to study new widths of the kagome strip \cite{Alvarez_Canfora_Reyes_Riquelme_2012} with generic $(q,v)$. 
The present work is a substantial improvement from CPM.

\section{Algorithm overview}
\label{seq_algorithm_overview}
\subsection{Data structure}
The definition of $G$ from Section \ref{seq_preliminaries_related_work} (see Figure \ref{fig_strip_lattice_general}) 
will be used in this section to explain the input data structure needed by the algorithm. Since the graph is a strip lattice, only layer $K_n$ of the graph $G$ is explicitly needed. 
The following naming scheme is now introduced for distinguishing two types of boundary vertices in the layer: \textit{shared vertices} and \textit{external vertices}.
For convention, \textit{shared vertices} are indexed top-down from $0$ to $m-1$ and correspond to the left-most ones of $K_n$, which are being shared with layer $K_{n-1}$. 
\textit{External vertices} are the right-most ones of $K_n$ and are indexed bottom-up 
from $|V|-m$ to $|V|-1$. Figure \ref{fig_data_structure} illustrates the data structure for an square strip of $m=3$.
\begin{center}
	\includegraphics[scale=0.5]{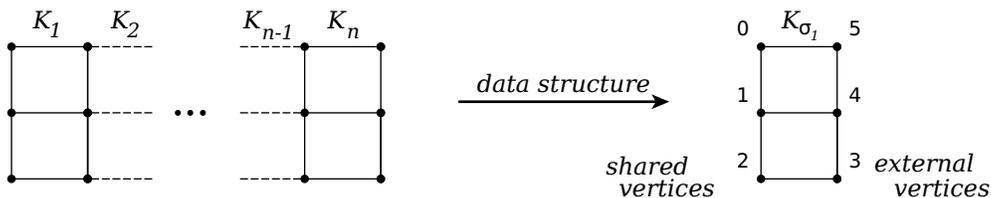}
	\captionof{figure}{Example data structure for a square lattice of width $m=3$.}
	\label{fig_data_structure}
\end{center}
\subsection{DC-based transfer matrix computation}
When using $(q,v)$ polynomials, the configuration space of generic $q$ transfer matrix methods turns out to be the set of all \textit{non-crossing partitions} on a sequence of $m$ serially 
connected vertices. The size of this configuration space is defined by the Catalan numbers:
\begin{equation}
\Gamma(m) = C_m=\frac{1}{m+1}\binom{2m}{m}=\frac{(2m)!}{(m+1)!m!}=\prod_{k=2}^m\frac{m+k}{k}
\end{equation}
We will first explain how the transfer matrix can be built from partial DC repetitions and then proceed to the \textit{parallel family trees} strategy. 

At this point we introduce two terminologies that are important for the rest of the section; \textit{initial configurations} and 
\textit{terminal configurations}. These configurations define a combinatorial sequence of identifications\footnote{For identification we mean a pair of vertices that actually 
represent a single vertex (they are identified). 
Graphically, it is represented by a crossed curved connecting the pair of vertices.} on the \textit{external} and \textit{shared vertices} of layer $K_n$. 
\textit{Initial configurations}, denoted $\sigma_i$ with $i \in [0..C_m-1]$, define a combinatorial sequence of identifications just on the \textit{external vertices} of $K_n$. The 
\textit{terminal configurations}, denoted $\varphi_j$ with $j \in [0..C_m-1]$, define a combinatorial sequence of identifications just on the \textit{shared vertices} of $K_n$. Initial 
configurations generate terminal ones, through the DC method.

The case of ${\sigma_1}$ is the basic case and matches $K_n$. That is, $\sigma_1$ 
is the \textit{initial configuration} where no identifications are applied to the \textit{external vertices} of $K_n$. 
It is equivalent as saying that $\sigma_1$ is the empty partition of the Catalan set. 
Similarly, $\varphi_1$ corresponds to the base case where no \textit{shared vertices} are identified. In other words, $\varphi_1$ is the empty configuration for the Catalan 
set on the \textit{shared vertices} of $K_n$. 
For illustration, Figure \ref{fig_configuration_space} shows the configuration space for the square lattice of width $m=3$:
\begin{center}
	\includegraphics[scale=0.5]{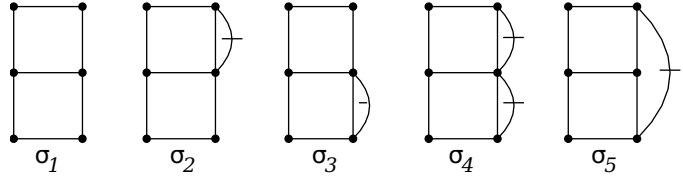}
	\captionof{figure}{The configuration space for a square lattice of width $m=3$.}
	\label{fig_configuration_space}
\end{center}
In order to compute the transfer matrix $M$ (row by row), one must apply $C_m$ 
partial DCs, each time to a different \textit{initial configuration} ${\sigma_i}$. 
Each one of the $C_m$ partial DC applications generates a row of $M$ in the form of partial partition functions on $(q,v)$, distributed into a maximum of $C_m$ \textit{terminal configurations}. 
By \textit{partial DC} we mean to perform DC on the layer, with the corresponding \textit{initial configuration} $\sigma_i$ applied, but stopping the recursion branches 
whenever they meet and edge that connects two \textit{shared vertices}. The 
stop condition on the recursion branches is needed otherwise one would be processing vertices 
and edges of the next layer of the strip, breaking the idea of a transfer matrix. 
For the example of Figure \ref{fig_data_structure} with $m=3$, the partial DC is applied to 
$\sigma_1$, $\sigma_2$, $\sigma_3$, $\sigma_4$ and $\sigma_5$ from Figure \ref{fig_configuration_space}. 

An example of a partial DC for the example of $m=3$ is illustrated in Figure \ref{fig_example_recdc} for the case when computing the first row. 
The process is analogous for the other four rows of $M$ (\textit{i.e.,} $\sigma_2$, $\sigma_3$, $\sigma_4$ and $\sigma_5$).
\begin{center}
	\includegraphics[scale=0.4]{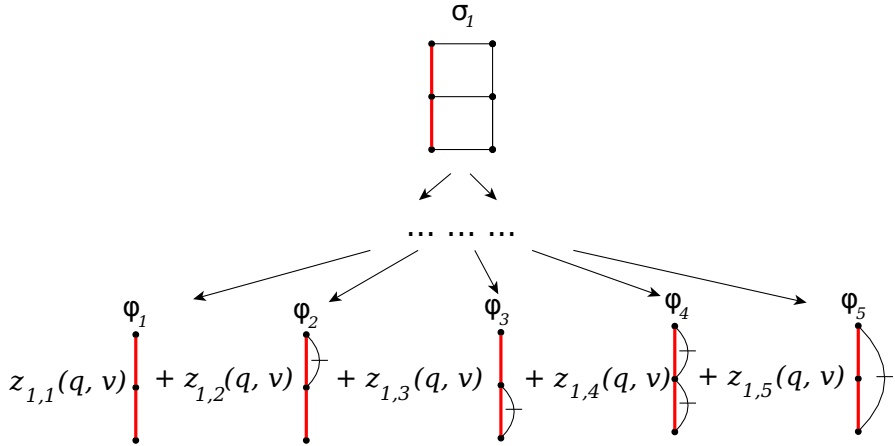}
	\captionof{figure}{Terminal configurations generated from a partial DC on a square strip of width $m=3$.}
	\label{fig_example_recdc}
\end{center}
Once a recursion branch has been stopped, partial partition functions $z_{i,j}(q,v)$ appear associated to remanents of the graph layer. 
Remanents are parts of the graph layer that cannot be computed (\textit{i.e}, edges connecting \textit{shared vertices}) and they match one of the $C_m$ possible \textit{terminal configurations} 
that can exist. For some \textit{initial configurations}, not all \textit{terminal configurations} may be generated from a single DC, but only a subset of them. 

A \textit{terminal configuration} $\varphi_j$ contains a unique sequence of planar identifications on the \textit{shared vertices} that is useful to differentiate one from another. 
We use the term \textit{key} to denote such sequences since they allow fast search and modification in a hash table. 
Proper construction of \textit{keys} are achieved by using a simple algebra 
that defines how multiple identifications on \textit{shared vertices} are combined. A key of $n$ identifications is denoted as 
$\Pi = \pi_{x_1,y_1} + \pi_{x_2, y_2} + ... + \pi_{x_n,y_n}$. The following properties hold true for keys:
\begin{align}
\pi_{a,b} & = \pi_{b,a}
\label{eq_algebra_0}\\
\pi_{a,b}+\pi_{c,d} & = \pi_{c,d} + \pi_{a,b}
\label{eq_algebra_1}\\
\pi_{a,b}+\pi_{b,c} & = \pi_{a,b,c}
\label{eq_algebra_2}
\end{align}
Properties (\ref{eq_algebra_0}) and (\ref{eq_algebra_1}) allow the application of a lexicographical order on the keys, 
while property (\ref{eq_algebra_2}) allows to combine them using transitivity. 
There are important differences when comparing this algebra to the partition algebras studied by Halverson and Ram \cite{halverson2005partition}, specially because the former is much simpler and defines operations on a single layer of points, while the latter defines a different set of operations for a partition monoid that is represented as a graph of two layers of points. Nevertheless, we can still find a relation with the number of partitions in the case of the planar sub-monoid $P_k$, which is $C_{2k}$ for two layers of length $k$, and the number of \textit{keys} for a single layer of length $m$, which is $C_m$.

Using Stirling's approximation, we have that $C_m \approx \frac{4^m}{m^{3/2}\sqrt{\pi}}$, which is consistent with the upper bound:
\begin{equation}
C_m = \frac{1}{m+1} {2m \choose m} \le {2m \choose m} \le 4^m
\label{eq_catalan_upperbound}
\end{equation}
Dutton and Brigham proved in 1986 that the Stirling approximation of the Catalan numbers is in fact already a valid upper bound \cite{Dutton:1986:CEB:10987.10992}. 
In addition, they obtain tighter lower and upper bounds for the Catalan numbers.
The cost of the DC-based transfer matrix method is the product of the cost of the partial 
DC and the size of the configuration space $C_m$. 

So far, the worst case running time of 
the algorithm for computing $M$ is:
\begin{equation}
T(G(V,E), m) = O\Big(\Gamma(m) \cdot DC(K_n))\Big) = O\Big(4^m\cdot min\Big(2^{|E'|}, \frac{1+\sqrt{5}}{2}^{|V'|+|E'|}\Big)\Big)
\label{eq_transfer_matrix_complexity}
\end{equation}

In the following sub-section, we show how a finer analysis 
can lead to a smaller configuration space of $\Gamma(m) = O(3^m)$ for computing a compressed transfer 
matrix $M$.

\subsection{Family trees strategy}
It is possible to reduce the Catalan configuration space by exploiting a symmetry present in 
the \textit{deletion-contraction} (DC) method, resulting in an exponentially faster 
algorithm. Basically, the idea is the following: \textit{if the DC procedure is forced 
to act first on certain external edges of the layer, 
and act later on the rest of the graph, then symmetries appear between nodes of the recursion 
tree and other initial configurations}. 
Exploiting such symmetry allows one to group many Catalan configurations into families of configurations, 
where a single DC procedure applied to the root node of a family contributes to the solution of the whole family.

Forcing DC to start on the external edges results in a recursion tree composed of two phases; 
(1) a perfect binary tree (PBT) of height $h=m-1-b$ and (2) several sub-trees $t_j$ with $j \in [1..2^{h}]$ (see Figure \ref{fig_dc_pbt}).
\begin{center}
	\includegraphics[scale=0.8]{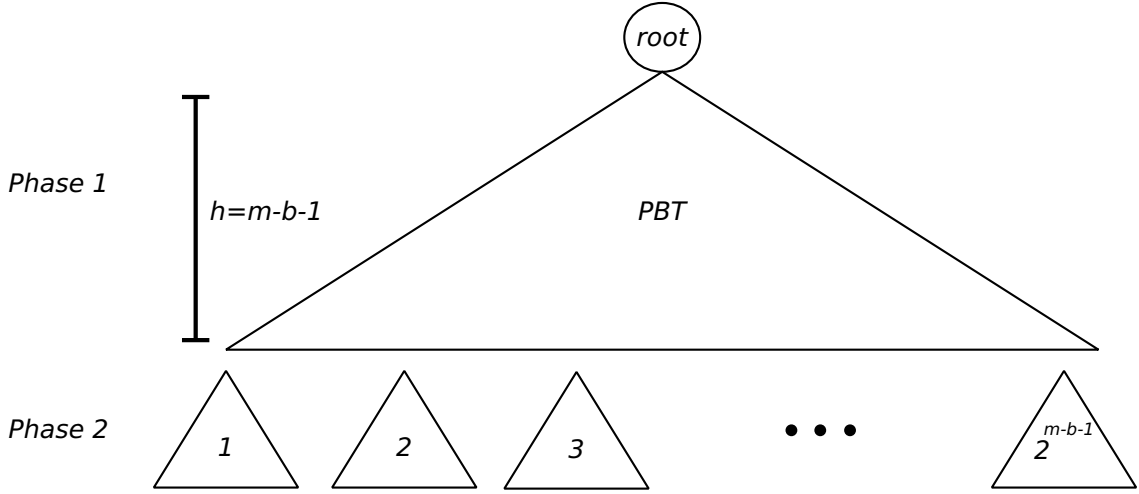}
	\captionof{figure}{When DC is forced to start on the external edges, the recursion is divided into two phases.}
	\label{fig_dc_pbt}
\end{center} 
Variable $b$ is the number of external edges that sit in between an identification $\pi_{ij}$ where at 
least one of its vertices is $i$ or $j$. These $b$ edges are left for phase (2) 
because they do not produce the symmetries needed for the \textit{family trees strategy}. 
Each node of the PBT of phase (1) that comes from a contraction produces 
a unique algebraic symmetry to one of the configurations found in the original Catalan set. 
The configuration of a contracted node from the recursion tree is denoted $\chi_i$ and the symmetric 
correspondence is $\chi_i \longleftrightarrow \sigma_i$.
All $\chi_i$ configurations that share the same PBT, together form a \textit{family tree}. 
Following the example of the square strip with $m=3$, its configuration space would be grouped into two \textit{family trees} (see Figure \ref{fig_example_family}); $\{\sigma_1, \sigma_2, \sigma_3, \sigma_4\}$ and $\{\sigma_5\}$, being 
$\sigma_1$ and $\sigma_5$ their root configurations, respectively.

\begin{center}
	\includegraphics[scale=0.4]{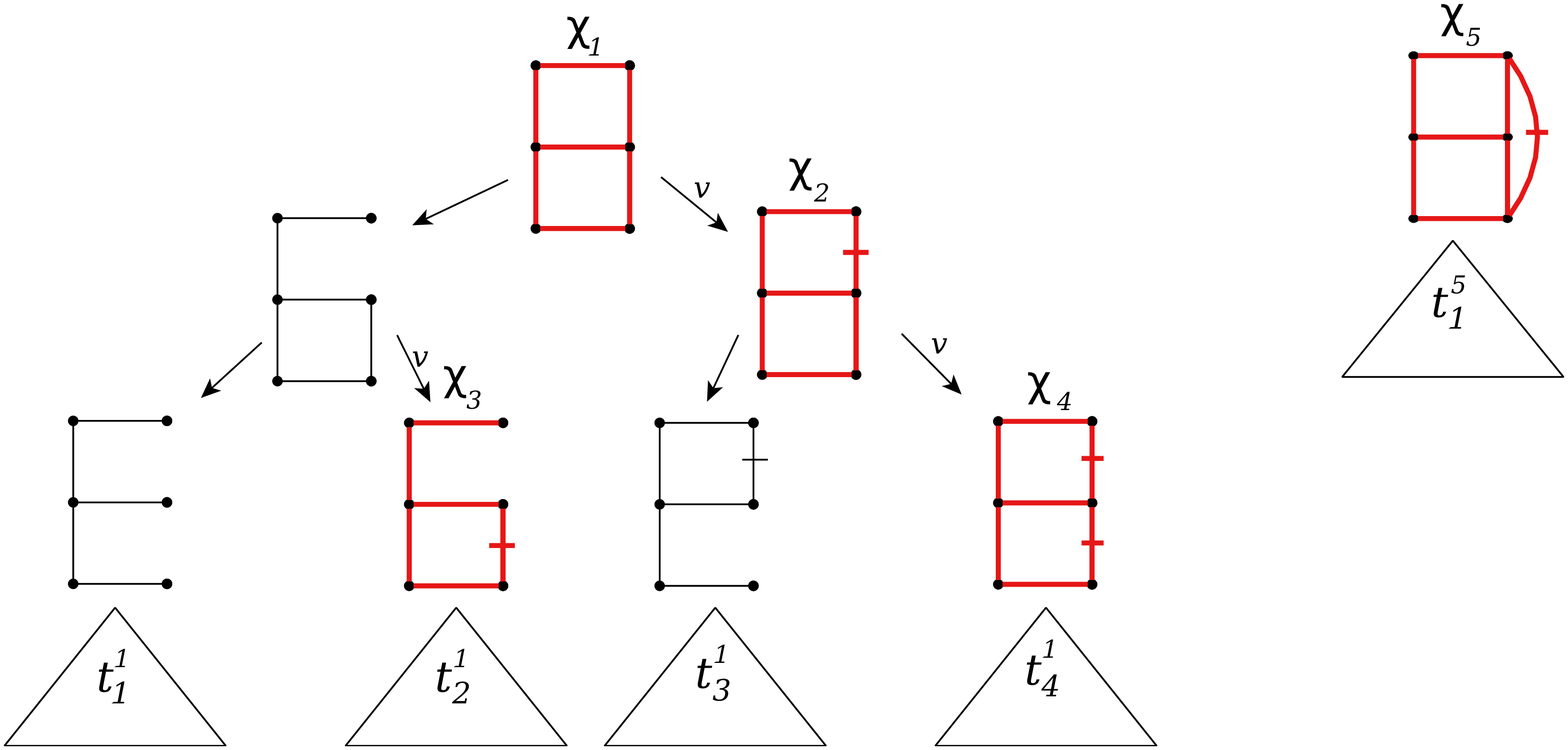}
	\captionof{figure}{An example of the perfect binary tree and subtrees for $m=3$.}
	\label{fig_example_family}
\end{center}

The solution of a configuration, namely $\langle {\sigma_i} \rangle$, is defined in terms of its symmetric $\chi_i$ found in the PBT:
\begin{equation}
\langle {\sigma_i} \rangle = (1+v)^c \sum_{k=0}^{2^d - 1} {v^{b(k)} \langle \chi_{i}^k} \rangle
\label{eq_correspondence}
\end{equation}
Variable $d$ denotes the number of deletions (\textit{i.e.}, holes in the external layer) and variable $c$ the 
contractions accumulated along its path, both starting from the root. 
The $(1+v)^c$ coefficient corresponds to the expression for the $c$ loops that are present in the external 
layer of $\sigma_i$, but are missing in $\chi_i$. For the example of the square strip of width $m=3$, $c=0,1,1,2,0$ for 
$\chi_1, \chi_2, \chi_3, \chi_4, \chi_5$, respectively.
Function $b(k)$ counts the number of non-zero bits of $k$ and the expression $\chi_{i}^k$ is the application 
of the binary mask $k$ just on the holes of $\chi_{i}$. The mask works as follows: if bit $k_j = 1$, 
with $j \in [0..d-1]$, then the $j$-th hole is filled with an edge, otherwise it is left as a hole.
 
When $d=0$, $\chi_i$ represents exactly the starting point 
of an eventual solution $\langle \sigma_i \rangle$, algebraically symmetric in $(1+v)^c$. When $d > 0$, 
$\chi_i$ is no longer the starting point of $\langle \sigma_i \rangle$, 
but instead it is the left-most node in an eventual recursion tree of the solution $\langle \sigma_i \rangle$, 
at level $d$. In order to compute $\langle \sigma_i \rangle$, $2^{d} - 1$ 
variations of $\chi_i$ are needed to build the missing steps and eventually reach $\sigma_i$ in a bottom-up way. 
An important property of the variations of $\chi_i$ is that they actually correspond to other family 
members within the PBT that will be eventually solved too. This means that there is no need to compute these 
variations, instead one has to make the correct relations between the 
different family members. We propose a hash map of the type $(\chi_i, r[])$ so that for each $\chi_i$, represented by its unique key, there is an array of related configurations $r[]$ that 
need $\langle \chi_i \rangle$. Each time a contracted configuration is reached in the PBT, 
equation (\ref{eq_correspondence}) is applied and 
$2^d - 1$ relations are inserted in the hash map. 
Figure \ref{fig_example_family_x3} illustrates the example 
of the strip of width $m=3$ when processing $\chi_3$; it needs $\chi_4$ 
in order to build the solution $\langle \sigma_3 \rangle$. 
\begin{center}
	\includegraphics[scale=0.4]{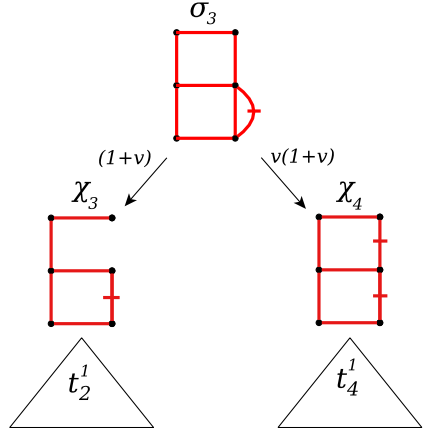}
	\captionof{figure}{An example of how $\chi_3$, with $d=1$, builds the solution of $\sigma_3$ with the help of $\chi_4$.}
	\label{fig_example_family_x3}
\end{center}

The solution for each family member $\langle \chi_{i} \rangle$ can be written in terms of the solutions of the $2^h$ subtrees. 
A convenient way for storing the solution for a whole family is to write a system of equations, using a linear combination of the $2^h$ sub-trees. A 
$v^c$ coefficient is included, where $c$ is the amount of contractions found in the path from the familiar 
to the sub-tree. For the example of the strip of $m=3$, the solution for the family of $\sigma_1$ is:
\begin{alignat}{2}
\langle\sigma_1\rangle &= \langle\chi_1\rangle & &= \langle t_1^1 \rangle + v\langle t_2^1 \rangle + v\langle t_3^1 \rangle + v^2\langle t_4^1 \rangle,\\
\langle\sigma_2\rangle &= (1+v)\langle\chi_2\rangle & &= (1+v)[\langle t_3^1 \rangle + v\langle t_4^1 \rangle]\\
\langle\sigma_3\rangle &= (1+v)[\langle\chi_3\rangle + v\langle\chi_4\rangle] & &= (1+v)[\langle t_2^1 \rangle + v\langle t_4^1 \rangle]\\
\langle\sigma_4\rangle &= (1+v)^2\langle\chi_4\rangle & &= (1+v)^2\langle t_4^1 \rangle
\end{alignat}
Note how $\langle\sigma_3\rangle$ includes $\langle \chi_4 \rangle$, as shown in Figure \ref{fig_example_family_x3}. The solution for the family of $\sigma_5$ is:
\begin{alignat}{2}
\langle\sigma_5\rangle &= \langle\chi_5\rangle & &= \langle t_1^5 \rangle
\end{alignat}
These equations, plus the solutions of the sub-trees, conform the compressed transfer matrix for the 
example strip of width $m=3$. It is important to mention that the sub-trees are stored only once and 
the system of equations use indices to the sub-trees.


Given how DC works, identification can only occur on pairs of vertices that are neighbors. This aspect of DC allows us to establish a formal definition for a family. 

\begin{mydef}
A family is a set of configurations in which for any chosen pair $\sigma_i$ and $\sigma_j$ of the set, the difference of their corresponding keys 
$\Pi^i$ and $\Pi^j$ is $\Pi^{i-j} = \pi_{x_1,x_1 + 1} + \pi_{x_2, x_2 + 1} + ... + \pi_{x_n,x_n +1}$.
\end{mydef}
In other words, the difference between $\sigma_i$ and $\sigma_j$ must 
only consist of identifications of length $l = 1$. Configurations that differ at least by one identification of length $l > 1$ belong to a different family. Each family is identified by 
its root configuration, therefore it is important to know which configurations are root and which are not.
\begin{mydef}
A root configuration is an instance of $K_n$ where its key $\Pi = \pi_{x_1,y_1} + \pi_{x_2, y_2} + ... + \pi_{x_n,y_n}$ satisfies $|x_i - y_i| > 1$ for $i \in [1..n]$.
\label{def_root_configurations}
\end{mydef}
That is, a root configuration is one that does not have identifications of length $l = 1$.
The number of root configurations will be denoted $\Delta_m$ as a function of the width $m$. 
We formulate the following expression for $\Delta_m$, based on Definition \ref{def_root_configurations} and using the \textit{inclusion-exclusion} principle:
\begin{equation}
\Delta_m = \sum_{k=0}^{m-1} (-1)^k {{m-1}\choose{k}} C_{m-k}
\label{eq_deltam}
\end{equation}

\begin{mylem}
\label{lemma_upperbound}
The amount of root configurations is upper bounded as $\Delta_{m} = O(3^m)$.
\end{mylem}
\begin{proof}
Using (\ref{eq_catalan_upperbound}) into (\ref{eq_deltam}) leads to the following bound:
\begin{alignat}{2}
	\Delta_m = \sum_{k=0}^{m-1} (-1)^k {m-1 \choose k} C_{m-k} \le \sum_{k=0}^{m-1} {m-1 \choose k} (-1)^k 4^{m-k} 	&= 4 \sum_{k=0}^{m-1} {m-1 \choose k} (-1)^k 4^{m-1-k}\\
\label{eq_delta_paso_binomial}																						&= 4 (4-1)^{m-1}\\
																													&= O(3^m)
	\label{eq_delta_m_4m}
\end{alignat}
Step \ref{eq_delta_paso_binomial} is obtained by using the Binomial formula with $x=4$ and $y=-1$. 
\end{proof}
The number of root configurations $\Delta_m$ corresponds to the number of \textit{non-crossing non-nearest-neighbor partitions} (\textit{nc-nnn}).
The number of \textit{nc-nnn} 
can also be counted with the Motzkin number evaluated at $m-1$; $\Delta_m = M_{m-1}$, where $M_m$ is:
\begin{equation}
M_m = \sum_{j=0}^{\lfloor m/2 \rfloor} {{m}\choose{2j}} C_j
\end{equation}
The asymptotic number of \textit{nc-nnn} partitions has been previously studied by Chang \textit{et. al.} in \cite{salas2002} by using the asymptotic behavior of $M_m$:
\begin{equation}
M_m = \frac{3^{3/2}}{2\sqrt{\pi}\ m^{3/2}}3^m\Big[1+O(m^{-1})\Big]
\end{equation}
Although the asymptotic bound was already obtained in two earlier works 
\cite{salas2002, salas_sokal_2011} in the context of \textit{nc-nnn partitions}, the 
proof of Theorem \ref{lemma_upperbound} still remains interesting as a short and 
alternative way to establish the $O(3^m)$ upper bound coming from an inclusion-exclusion 
formulation that has not considered the Motzkin numbers. 
\subsubsection{Upper bound for relating $k$-hole familiars}
Counting the amount of family relations within a DC procedure allows one to precise an upper bound on the number of accesses made to the hash map. For each DC application, the cost of 
relating family members is defined as:
\begin{equation}
g(h) = \sum_{k=0}^{h-1} c(k,h) r(k)
\label{eq_gfunc}
\end{equation}
Where $r(k) = 2^k - 1$ is the cost of performing the relations for a $k$-hole configuration. Function $c(k,h)$ counts the number of $k$-hole configurations, which is a subset of the 
total number of familiars. Since familiars can only be contracted nodes within the PBT, the size of a family is $2^{h-1}$. A direct upper bound can be computed assuming the worst case for 
$r(k)$:
\begin{equation}
g(h) < (2^m - 1) \sum_{k=0}^{h-1} c(k,h) \le (2^m - 1)2^h < 4^m = O(4^m)
\end{equation}
A tighter upper bound is possible when $c(k,h)$ is analyzed more carefully. The following pattern can be found when counting the number of $k$-hole configurations.
\begin{alignat}{2}
c(0,h) &= h\\
c(1,h) &= 1+2+...+h-1\\
c(2,h) &= (1)+(1+2)+...+(1+2+3...+h-2)\\
c(3,h) &= \big[(1)\big] + \big[(1) + (1+2)\big] + ... + \big[(1) + (1+2) + ... + (1+2+3+ ... +h-3)\big]
\end{alignat}
The recursion for $c(k,h)$ is:
\begin{alignat}{1}
c(k,h) &= \sum_{i=0}^{h-k} c'(k-1, i),\ \ \ 1 \le k \le h-1\ \ \&\ c(0,h) = h\\
c'(k,h) &= \sum_{i=0}^{h} c'(k-1, i), \ \ \ \ 1 \le k \le h-1\ \ \&\  c'(0,h) = h
\end{alignat}

Function $c(k,h)$ is equivalent to counting the number of $k$-faces in a 
regular $(h-1)$-simplex \cite{coxeter1973regular}. A regular $(h-1)$-simplex is a $(h-1)$-dimensional 
polytope that is the convex hull of $h$ vertices in a regular spatial distribution. 
A regular simplex can also be seen as the generalization of the notion of a triangle or a 
tetrahedron, for an arbitrary dimension. A regular $(h-1)$-simplex can be drawn in the plane by placing 
$h$ vertices inscribed in a circle, with all pairs connected (see Figure \ref{fig_simplex}). 
\begin{center}
	\includegraphics[scale=0.7]{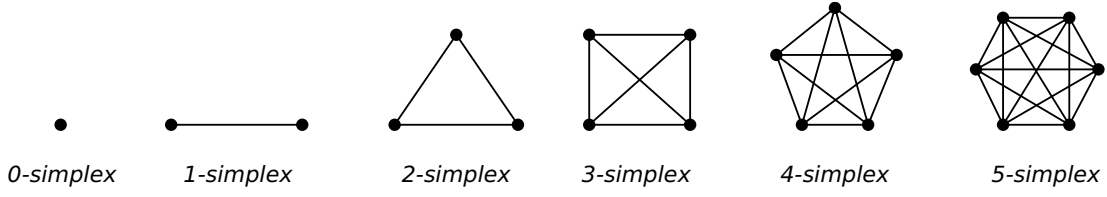}
	\captionof{figure}{Examples of regular simplexes drawn on the plane.}
	\label{fig_simplex}
\end{center}
The number of $k$-faces in a $(h-1)$-simplex \cite{Henk:1997:BPC:285869.285884} is 
defined as:
\begin{equation}
c(k,h) = {{h}\choose{k+1}}
\label{eq_simplex}
\end{equation}
Using (\ref{eq_simplex}) in (\ref{eq_gfunc}), we have that
\begin{equation}
g(h) = \sum_{k=0}^{h-1} {{h}\choose{k+1}} (2^k - 1)
\end{equation}

\begin{mylem}
The cost of relating all configurations within a PBT is upper bounded as $g(m-1) = \frac{1}{6}(3^m - 3\cdot2^m + 3) = O(3^m)$.
\end{mylem}
\begin{proof}
For simplicity, we will assume that every DC application processes the default 
initial configuration. This configuration is the one that spans the largest family, 
hence the worst case where $b = 0$, that is $h=m-1$.
\begin{equation}
g(h) \le g(m-1) = \sum_{k=0}^{m-2} {{m-1}\choose{k+1}} (2^k - 1) = \sum_{k=0}^{m-2} {{m-1}\choose{k+1}}2^k - \sum_{k=0}^{m-2} {{m-1}\choose{k+1}}
\label{eq_proof_relations00}
\end{equation}
Both summations obey the following form:
\begin{alignat}{2}
\sum_{k=0}^{m-2} {m-1 \choose k+1} a^k 	&= 	\frac{1}{a} \sum_{k=1}^{m-1} {m-1 \choose k} a^k =  \frac{1}{a} \left(-1 + \sum_{k=0}^{m-1} {m-1 \choose k} a^k\right)
\end{alignat}
Using the Binomial theorem for the summation, we get
\begin{alignat}{2} 
\frac{1}{a} \left(-1 + \sum_{k=0}^{m-1} {m-1 \choose k} a^k\right) &=  \frac{(a+1)^{m-1}-1}{a}
\end{alignat}
Using $a=2$ and $a=1$ leads to the first and second terms of Eq. (\ref{eq_proof_relations00})
\begin{equation} 
g(h) \le g(m-1) = \frac{3^{m-1}-1}{3} - \frac{2^{m-1}-1}{2} = \frac{1}{6}(3^m - 3\cdot2^m + 3) = O(3^m)
\end{equation} 

\end{proof}

\subsubsection{Running time of the family trees strategy}
The asymptotic sequential running time of the family trees algorithm applied to a layer $K(V',E')$ of a strip lattice is:
\begin{alignat}{2}
T(m, K(V',E')) 	&= \Delta_m \Big(DC + g(m-1)\Big)\\
				&= O\Big(3^m\Big(min\Big(2^{|E'|}, \frac{1+\sqrt{5}}{2}^{|V'|+|E'|}\Big) + 3^m\Big)\Big)
\end{alignat}

The extra cost provided by $g(m-1)$ does not incur in too much extra computation compared to the cost of DC itself, where the amount of edges of $K(V',E')$ must at least 
double the amount of edges in the boundary, that is $E' \ge 2(m-1)$. Additionally, $g(m-1)$ is considering the worst case for each root configuration where $h = m-1$. In practice, 
all configurations, except for the default one, will have $h < m-b-1$ with $b>0$.

\subsubsection{Parallel family trees}
By default, the algorithm does not know the $\Delta_m$ different root configurations except for ${\sigma_1}$ 
which is given as part of the input of the strip lattice and is the one that triggers the computation. Under this scheme, the configuration space would have to be explored incrementally, 
each time adding a sub-set of configurations from the \textit{terminal configurations} found from a DC application. This is indeed a problem for parallelization because the data-parallel elements are being discovered 
sequentially, limiting the efficiency and scalability of a parallel computation. In order to solve this problem, we use a recursive generator $g(A[\ ][\ ], s, H,S)$, that with the help of a hash table $H$, generates 
all the $\Delta_m$ configurations before hand and stores them in an array $S$. $A[\ ][\ ]$ is an auxiliary array that stores the intermediate auxiliary subsequences and $s$ is the accumulated 
sequence of identifications. 
Before the first call to $g(A[\ ][\ ], s, H,S)$, $A = [[0,1,2,...,m-1]]$, $s$ is null and $H$ as well as $S$ are empty. 
$g(A[\ ][\ ], s, H,S)$ is defined as:
\begin{lstlisting}[frame=single]
g(A[][],s,H,S){
    if(!add_sequence(s,H,S))
        return;
    for(int k=0; k<A.size(); k++){
        for(int j=2; j<A[k].size(); j++){
            for(int i=0; i<j-1; i++){
                if(can_identify(A[k],i,j)){
    	            cA = copy(A);
    	            cs = copy(s);
    	            identify(cA,i,j,k,cs);
    	            divide(cA,i,j,k);
    	            g(cA,cs,H,S);
}}}}}
\end{lstlisting}

Basically, $g(..)$ performs a recursive partition of the domain $A$. 
If $|j-i| \le 3 $ then no further identifications can be carried on, otherwise the identification would be of length $l = 1$ and the generated configuration would not be a \textit{root configuration}. 
Similarly, for the top and bottom parts if $|j-i| \le 2$ then no more identifications are possible. Each time a new identification $i,j$ is added, 
the resulting configuration is checked in the hash table. If it is a new configuration, then it is added, else it is 
discarded as well as all further recursion computations continuing from that point. By using this approach we ensure that redundant recursion branches are never computed.
Once $g(..)$ has finished, $S$ becomes the array of all possible configurations and $H$ the hash that maps configurations to indices. 

Parallel family trees are achieved by first generating all root configurations with $g(..)$, followed by the parallel computation of $p$ family trees simultaneously, using $p$ processors and a total of $\Delta_m/p$ family trees 
per processor. The initial \textit{key} needed by each processor $p_i$ is obtained by reading in parallel from $S[p_i]$, assuming the PRAM-CREW model. 
Once the \textit{key} is obtained, it is applied to the \textit{external vertices} of its own local copy of the base layer ${\sigma_1}$. 
Foster's four-step strategy \cite{fo_95} describes the design process of a parallel algorithm; \textit{partitioning, communication, agglomeration, mapping}.
The design steps for the parallel family trees is illustrated in Figure \ref{fig_foster_strategy_mp}.
\begin{center}
	\includegraphics[scale=0.3]{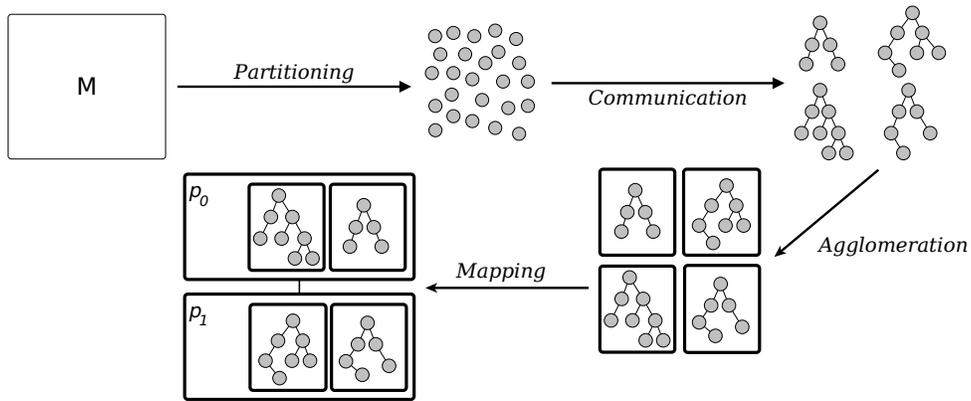}
	\captionof{figure}{Foster's four step strategy for achieving parallel family trees, for two processors.}
	\label{fig_foster_strategy_mp}
\end{center}
The work for each processor $p_i$ is divided in the following steps: (1) pick one root configuration \textit{key} 
from $S[]$, (2) apply it to its local copy of the $K_{\sigma_1}$ layer, (3) perform 
the DC procedure, (4) write the results into non-volatile memory, 
\textit{i.e.,} sub-tree results as well as the linear equations into disk, and (5) go to step (1) if there are still root configurations remaining. 
For step (3), familiars of a root configuration are detected at runtime within the PBT by computing its \textit{key}, each time the recursion comes from a contraction. When the beginning of a sub-tree is reached, 
no more familiars are guaranteed to be found on what is left of the recursion, therefore the algorithm can proceed to compute the whole sub-tree without needing to check for the existence of familiars. 
The solution of a sub-tree $t_i$ is a vector of expressions $z_{i,j}(q,v)$ that associates a $j$ index to a \textit{terminal configuration} $\varphi_j$ within the sub-tree $t_i$. 
The hash-map $H$ from the generator becomes useful for searching with average cost $O(1)$ the index $j$ of a terminal configuration $\varphi_j$. 
Also, $H$ ensures that all vectors are consistent with the order established in the generator and in the transfer matrix. 

The $2^{m-1}$ sub-tree vectors and the coefficients for the set of equations provide the solution for a whole family. Both of these results are saved locally for each processor. 
This output format based on sub-trees and coefficients makes the matrix compressed in the same proportion of the improvement in the running time. 

The asymptotic running time for the parallel family trees algorithm using $p$ processors is:
\begin{alignat}{2}
T(m) 	&= O\Big(\frac{3^m}{p} \Big[DC + g(k,m)\Big]\Big)\\
		&= O\Big(\frac{3^m}{p} \Big[min\Big(2^{|E'|}, \frac{1+\sqrt{5}}{2}^{|V'|+|E'|}\Big) + 3^m\Big]\Big)
\label{eq_pft_runningtime}
\end{alignat}

Further computations for achieving physical results require decompression of the matrix, leading to a matrix of Catalan dimensions again. 
In practice, large symbolic matrices need first to be evaluated before doing any analysis. 
If the numerical evaluation is performed before decompressing the matrix, then the process 
is much faster than first decompressing and then evaluating, even faster than evaluating an uncompressed transfer matrix on $(q,v)$. 
Numerical evaluation has the potential to be exponentially faster as a consequence of the parallel family trees compression, which is in the 
same order of the running time improvement.

The analysis of the algorithm has been made for the case of free boundary conditions but it is not restricted to it. For different boundary conditions such as cylindrical, full periodic or cyclic, the parallel family trees 
can be still applied following the same principle, while taking advantage of additional symmetries like the dihedral group in the cylindrical case. The rest of the paper assumes free boundary conditions unless we explicitly 
mention the contrary. 

For the case of a finite strip, the initial conditions vector $\vec{Z_1}$ is computed by applying DC to each one of the $C_m$ \textit{terminal configurations}:
\begin{equation}
\vec{Z_1} = (DC({\varphi_1}), DC({\varphi_2}), ..., DC({\varphi_{C_m}}))
\end{equation}
The computation of $\vec{Z_1}$ has very little impact on the overall cost of the algorithm and practically costs 
$O(mC_m)$ in time because a terminal configuration contains mostly \textit{spikes} and/or \textit{loops}, which are linear in cost for DC.

\section{Algorithm improvements}
\label{seq_algorithm_optimizations}
\subsection{Serial and Parallel paths}
The DC contraction procedure can be improved for graphs that present \textit{serial} or \textit{parallel} paths between two endpoints $v_a$ and $v_b$, as shown in Figure \ref{fig_opt_sp}.
\begin{center}
	\includegraphics[scale=0.7]{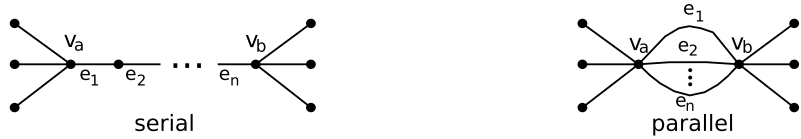}
	\captionof{figure}{Serial and parallel paths.}
	\label{fig_opt_sp}
\end{center}
A \textbf{\textit{serial path}}, denoted $s$, is a set of 
edges $e_1,e_2,...,e_n$ that connect sequentially $n-1$ vertices between $v_a$ and $v_b$.
It is possible to process a serial path of $n$ edges in one recursion step by using the following expression;
\begin{equation}
Z(K,q,v)=\Bigg[\frac{(q+v)^n - v^n}{q}\Bigg]Z(K_{-s},q,v) + v^nZ(K_{/s},q,v)
\end{equation}
\label{sec_optimizations}
A \textbf{\textit{parallel path} $p$} is a set of edges $e_1,e_2, ..., e_n$ that reduntandly connect $v_a$ and $v_b$.
It is possible to process a parallel path of $n$ edges in one recursion step by using the following expression;
\begin{equation}
Z(K,q,v)=Z(K_{-p},q,v) + \big[(1+v)^n - 1\big]Z(K_{/p},q,v)
\end{equation}

\subsection{Axial Symmetry}
One practical optimization is to detect the lattice's reflection symmetry when computing the root configurations as well as the Catalan configurations. 
When detecting reflection symmetry, the size of the configuration space is decreased for all symmetric pairs of \textit{configurations}, no matter if it 
is \textit{initial, terminal or root}. As the width of the strip lattice increases, the number of symmetric states increases too, leading to configuration spaces 
almost half the size of the original. 
We establish reflection symmetry between two \textit{configurations} $\varphi_a$ and $\varphi_b$ with 
keys $\pi_{a_1,...,a_n}$ and $\pi_{b_1,...,b_n}$ respectively in the following way:
\begin{equation}
\pi_{a_1,...,a_n} = \pi_{b_1,...,b_n} \Leftrightarrow a_i=(m-1)-b_{n-i+1} 
\end{equation}
Exploiting this symmetry results in a matrix size $C_m^s$:
\begin{align}
C_m^s 	& = \frac{C_m}{2} + \frac{m!}{2\lfloor\frac{m}{2}\rfloor!}
\end{align}
For large values of $m$, $C_m^s \approx \frac{C_m}{2}$.

For the case of \textit{root configurations}, Chang \textit{et. al.} \cite{salas2002} proved that the number of non-crossing non nearest-neighbor partitions under reflection symmetry, which we denote $\Delta_m^s$, is:
\begin{equation}
\Delta_m^s = \frac{1}{2}M_{m-1} + \frac{(m' - 1)!}{2} \sum_{j=0}^{\lfloor m'/2 \rfloor} \frac{m' - j}{(j!)^2(m'-2j)!}
\end{equation}
where $m' = \Big \lfloor \frac{m+1}{2} \Big \rfloor$. The expression was also obtained by Salas and Sokal \cite{salas_sokal_2011} for studying the square lattice symmetries 
when $v=-1$. When $m \to \infty$ we have: 
\begin{equation}
\Delta_m^s \sim \frac{\sqrt{3}}{4\sqrt{\pi}\ m^{-3/2}} 3^m \Big[1 + O({m^{-1}})\Big]
\end{equation}

Table (\ref{table_nonsym_sym_growth}) shows how the amount of Catalan and root configurations increase for non-symmetric and symmetric lattices up to $m=14$.
\begin{table}[ht!]
\begin{center}
\caption{Number of Catalan and root configurations under non-symmetric and symmetric cases.}
  \begin{tabular}{ r | r | r | r | r}
	$m$	& $C_m$	& $C_m^s$	&	$\Delta_m$	& $\Delta_m^s$ \\ 
	\hline
	1	& 1		& 	1	&	1	&	1\\
	2	& 2		& 	2	&	1	&	1\\
	3 	& 5 	& 	4	&	2	&	2\\
   	4 	& 14	& 	10	&	4	&	3\\
	5 	& 42	& 	26	&	9	&	7\\
	6	& 132	& 	76	&	21	&	13\\
	7	& 429	& 	232	&	51	& 	32\\
	8	& 1430	& 	750		&	127	&	70\\
	9	& 4862	& 	2494	&	323		&	179\\
	10	& 16796	& 	8524	&	835		&	435\\
	11	& 58786	& 	29624	&	2188 	&	1142\\
	12	& 208012	& 104468		&	5798 	&	2947\\
	13	& 742900 	& 372308		&	15511 	& 7889\\
	14	& 2674440	& 1338936		&	41835 	&	21051
  \end{tabular}
  \label{table_nonsym_sym_growth}
\end{center}
\end{table}
If cylindrical boundary conditions are used, then the reflection symmetry can be replaced by the symmetry of the dihedral group which 
further reduces the size of the matrix. For this manuscript we limit our work to the case of free boundary conditions.

\section{Implementation}
\label{sec_implementation}
We tried two implementations for the parallel family trees parallel algorithm; one using OpenMP \cite{Chapman:2007:UOP:1370966} and the other one using MPI \cite{mpi}. We observed that the MPI 
implementation achieved better performance in the multi-core scenario and allows parallel computation in a distributed scenario.
For this, we decided to continue the research with the MPI implementation for both multi-core and distributed scenarios. Basic mathematical operations on symbolic expressions are handled through the GiNaC C++ 
library \cite{Bauer20021}.
Parallel execution of the algorithm receives two parameters; the number of processors $p$ and the block size $B$, which is the amount of consecutive jobs per process. 
When the parallelization is unbalanced, the value of $B$ plays an important role for efficiently distributing work to all processors.
In our implementation we make each process to generate its own $H$ lookup table and $S$ array. This small sacrifice in memory leads to better performance than if $H$ and $S$ 
were shared among all processes. There are mainly three reasons why the replication approach is better than the sharing approach: (1) caches will not have to deal with consistency of 
shared data, (2) there is no sending/receiving of data structures and (3) the allocation of the replicated data is correctly placed on memory modules when working under a 
NUMA architecture. The last claim is true because on NUMA systems memory allocations on a given process are automatically placed in its fastest location according to the 
NUMA topology between memory and CPU cores. It is responsibility of the OS (or make manual mapping) to stick the process to the same processor throughout the entire computation.

The implementation writes each row to a persistent secondary memory (\textit{i.e.,} HDD or SSD) as soon as it is computed. Each processor does this with its own file, therefore the matrix is fragmented into $p$ files. 
In practice, a fragmented matrix is not a problem at all, because numerical evaluation is needed before using the matrix in its full form. Furthermore, a fragmented matrix 
allows parallel numerical evaluation.

\section{Performance results}
\label{sec_performance}
We have realized performance tests for the parallel transfer matrix method implemented with MPI for both shared and distributed memory scenarios.
The experimental design consists of measuring the main performance metrics (\textit{i.e.,} running time, speedup, efficiency, knee) 
of the implementation by computing the compressed transfer matrix several times, each time varying the number of processors $p$. We also 
compute the improvement factor with respect to previous work \cite{DBLP:conf/hpcc/NavarroHC13}. The experiments are divided 
into two categories; (1) multi-core and (2) cluster. For each case, we measure performance with two strip lattices; 
(1) \textit{square} and (2) \textit{kagome}, respectively (see Figure \ref{fig_test_lattices}). 
\begin{center}
\includegraphics[scale=0.6]{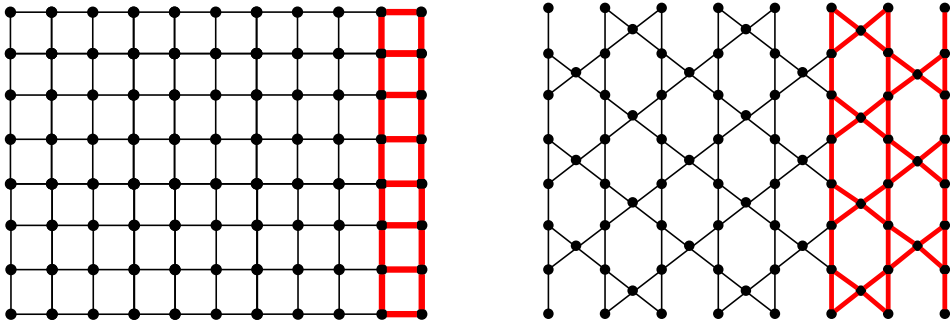}
\captionof{figure}{The square and kagome lattices used for measuring performance.}
\label{fig_test_lattices}
\end{center}
Explicit algebraic expressions for the sparse-matrix factorization of $M$ for all the Archimedean lattices (which include the square and kagome lattices)
have been computed by Jacobsen \cite{1751-8121-47-13-135001}, on finite lattice regions of up to $|E| = 882$ edges. The approach taken by the sparse-matrix 
differs from the standard transfer matrix technique, since the former processes a whole finite 
lattice region, using one sparse matrix computation per edge, while the latter 
computes a dense $TM$ for each different graph layer of width $m$.

\textit{Note}: $PFT$ refers to the actual \textit{Parallel Family Trees} strategy and $PCM$ to the \textit{Parallel Catalan Method} from \cite{DBLP:conf/hpcc/NavarroHC13}.
\subsection{Multi-core results} 
The machine used for the multi-core performance tests has an 8-core CPU AMD FX-8350 at $4.0$ GHz, 8GB of RAM and uses the \textit{openMPI} implementation of the MPI standard \cite{mpi}.

\subsubsection{Square strip lattice test}
For the square lattice, we measure performance for 9 different strip widths in the range $m \in [2,10]$. 
For each width, we measure 8 average execution times, one for each value of $p \in [1,8]$. As a whole, 
we perform a total of 72 average measurements for the square test. The standard error for each average execution time is below 5\%.
Different block sizes where tested, giving no significant difference on performance. For this reason, we kept a block size of $B=1$.
The other performance measures include speedup, efficiency and the \textit{knee}\footnote{In the knee, point counting is in reverse order.} \cite{DBLP:journals/tc/EagerZL89}. 
In this case we took advantage of the reflection symmetry for all sizes of $m$.

Figure \ref{fig_square_performance_results} shows all four performance measures for the square lattice. 
\begin{figure*}[ht!]
\centering
\begin{tabular}{cc}
\includegraphics[scale=0.59]{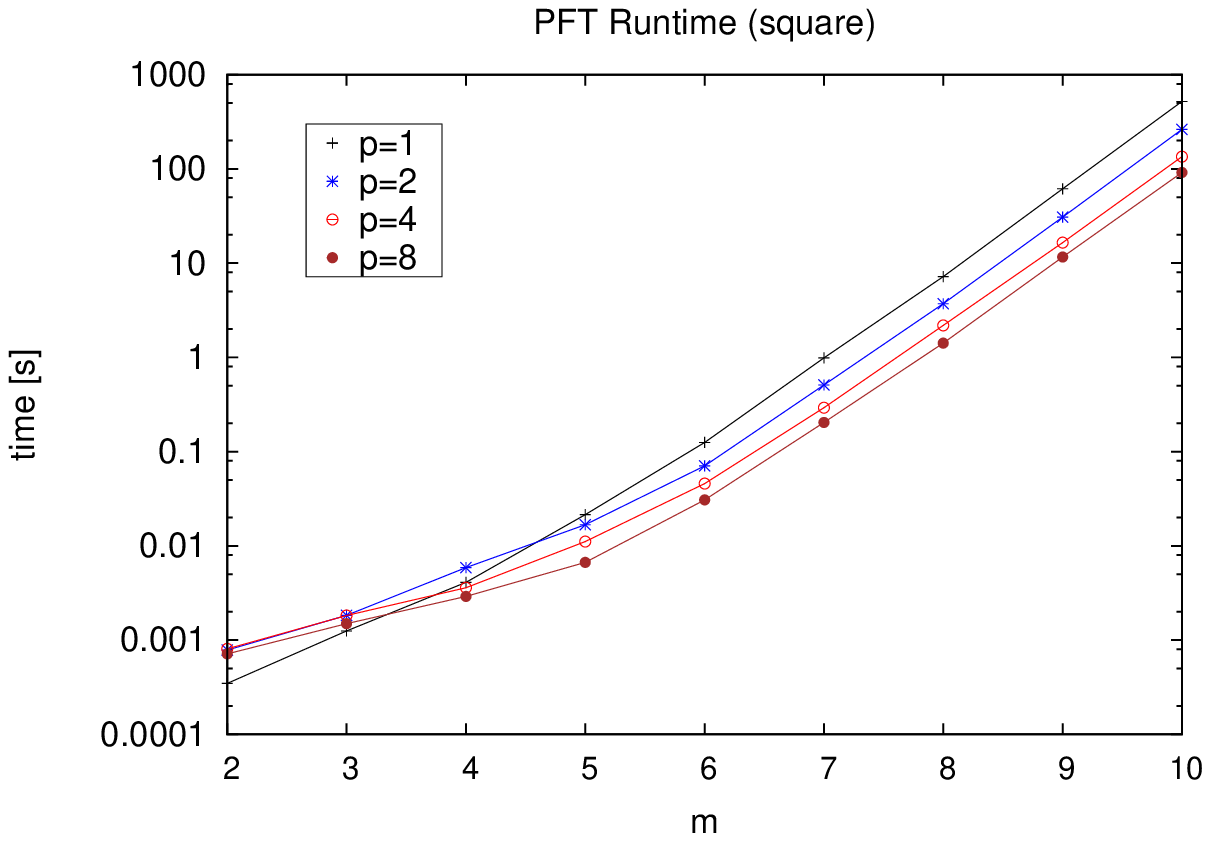} &
\includegraphics[scale=0.59]{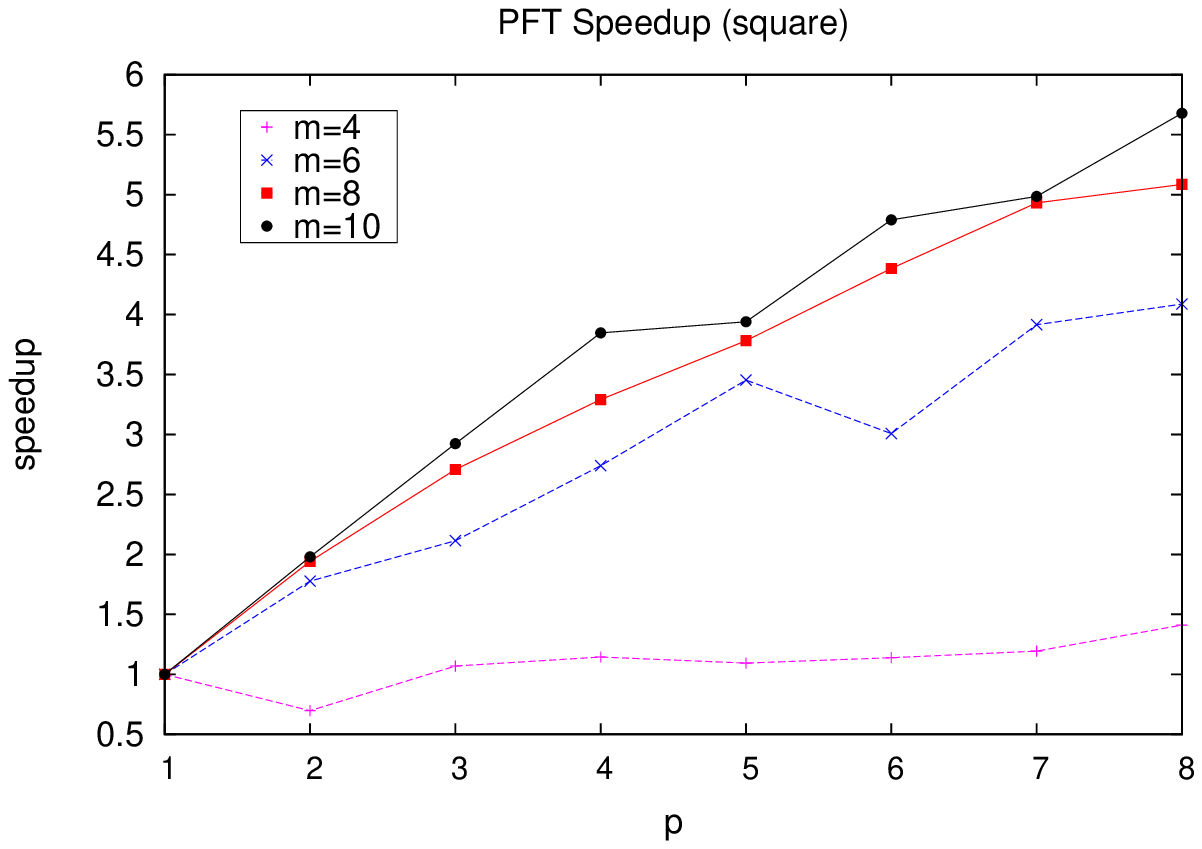}\\
\includegraphics[scale=0.59]{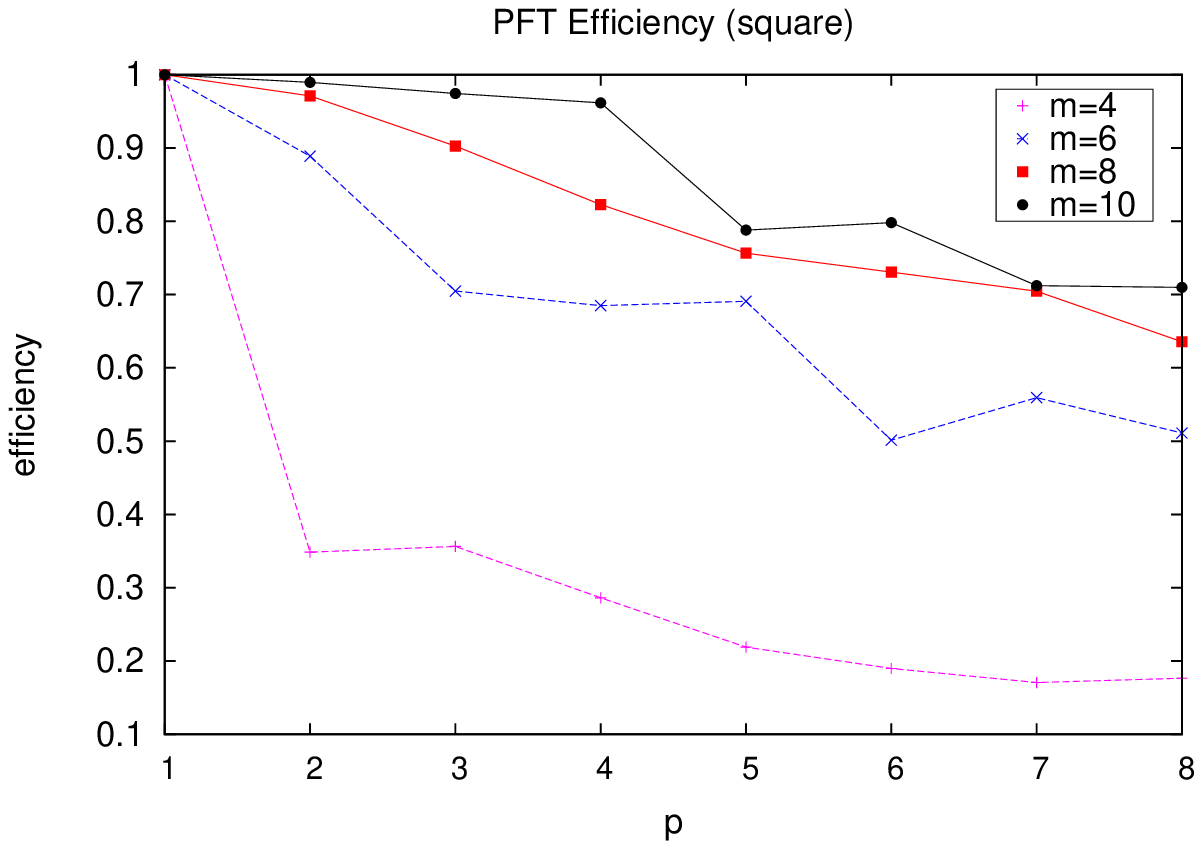} &
\includegraphics[scale=0.59]{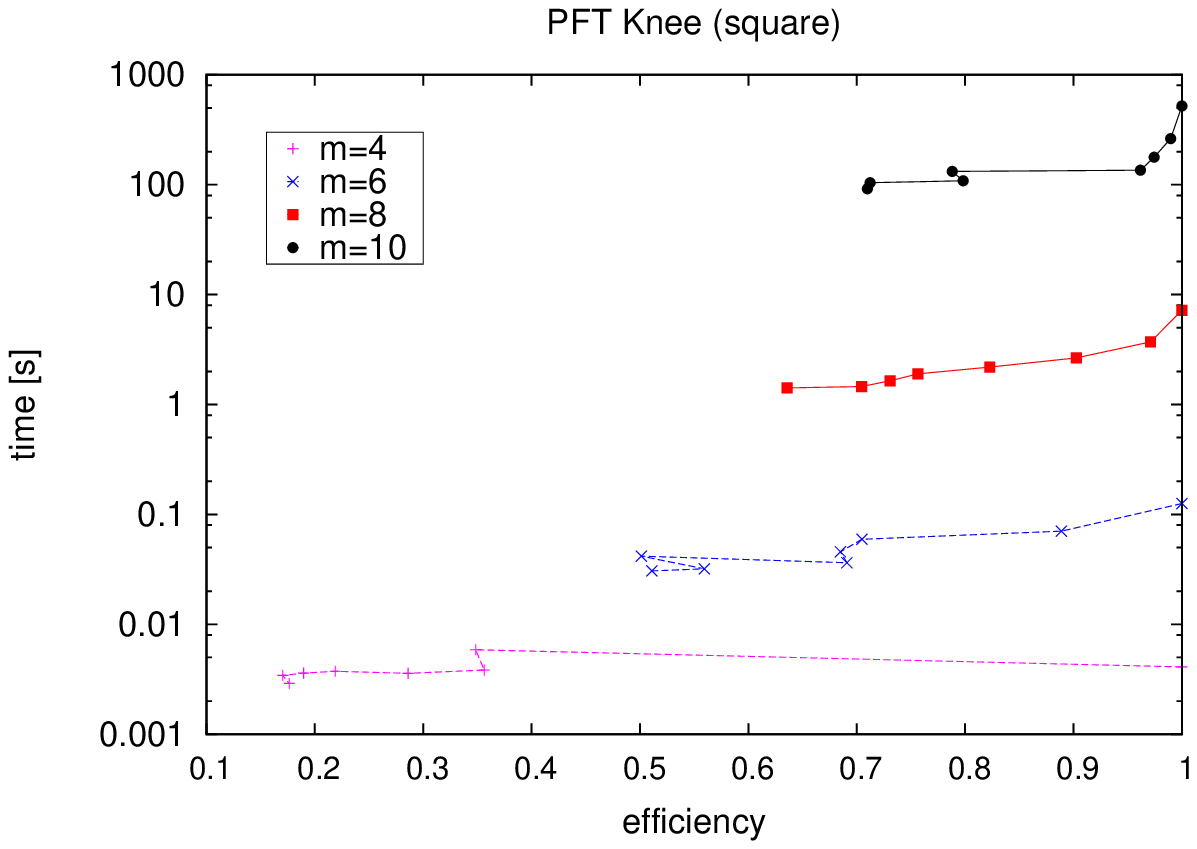}\\
\end{tabular}
\caption{Multi-core running time, speedup, efficiency and knee for the square strip test.}
\label{fig_square_performance_results}
\end{figure*}
From the results, we observe that the running time grows at an exponential rate 
which is compatible with the upper bound in (\ref{eq_pft_runningtime}), 
assuming that the cost of DC had a little impact on the algorithm. 
Indeed it is possible for DC to have a little impact, considering that algorithmic 
improvements are linear and they occur with more or less frequency depending on the 
edge selection order \cite{haggard_computing_tutte_polynomials} and the lattice structure. 
For the speedup, there is improved performance for every value of $p$ as long as $m > 4$. 
For $m \le 4$, the problem is not large enough to justify parallel computation, hence the 
overhead from MPI makes the implementation perform poorly and sometimes even worse than 
the sequential version. The plot of the execution times confirms this behavior since the curves cross each other for in the transition from $m = 3$ to $m = 4$.
The maximum speedup obtained was $5.7$ when using $p=8$ processors. From the lower left plot we can see that efficiency 
decreases as $p$ increases, which is expected in every parallel implementation. What is important is that for large enough problems (\textit{i.e.}, $m>6$), 
efficiency is over 62\% for all $p$. For the case of $p=4$, we report at least 95\% of efficiency, which is close to perfect linear speedup. 
For $m \le 6$, the implementation is not so efficient because the amount of computation involved is not enough to keep all cores working at full capacity. 
The \textit{knee} is useful for finding the optimal value of $p$ for a balance between efficiency and computing time. 
It is called knee because the hint for the optimal value of $p$ is located in the knee of the curve (thought as a leg), that is, its lower right part. In order to know the value of $p$ suggested by the knee, one has to count the position of the closest point to the knee region, in reverse order. Our results of the knee for $m>6$ show that the best balance of performance and efficiency is achieved with $p=4$ (for $m \le 6$, the knee is not effective since there was no speedup in the 
first place). In other words, while $p=8$ is faster, it is not as efficient as with $p=4$.

\subsubsection{Kagome strip lattice test}
For the test of the kagome lattice, we used 6 different strip widths in the range $m \in [2,7]$. 
For each width, we measured 8 average execution times, one for each value of $p \in [1,8]$. As a whole, 
we performed a total of 48 measurements for the kagome test. The standard error 
for each average execution time is below 5\%.
Additional performance measures such as speedup, efficiency and knee have also been computed. Different values of block size were tested, achieving noticeable differences on 
performance as $B$ changed. We found by experimentation that $B=1$ makes the work assignment slightly more balanced. In this test we can only use lattice axial symmetry for 
$m = 2, 4, 6, 8, ...\ $. For this reason we decided to run the whole kagome benchmark without axial symmetry in order to maintain a coherence between odd and even values of $m$.

Figure \ref{fig_kagome_performance_results} shows the performance results for the kagome strip test.
\begin{figure*}[ht!]
\centering
\begin{tabular}{cc}
\includegraphics[scale=0.59]{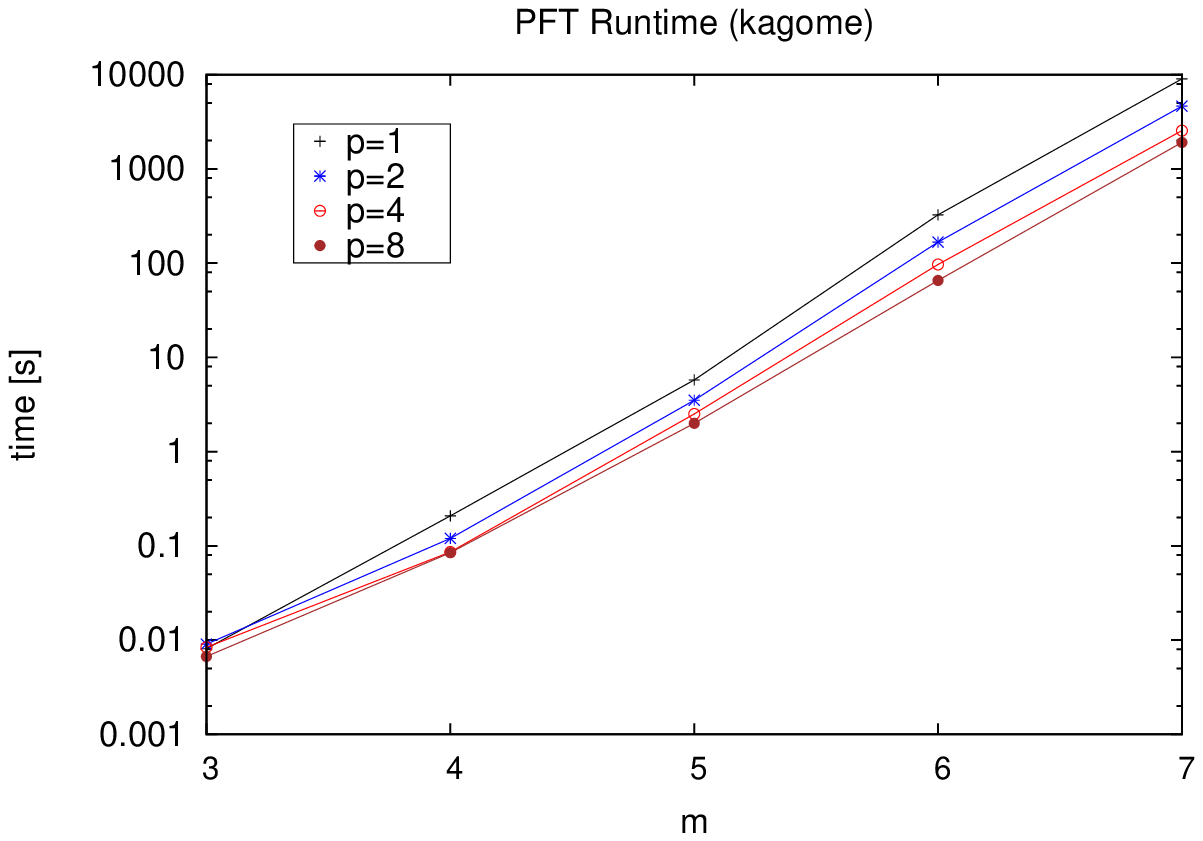} &
\includegraphics[scale=0.59]{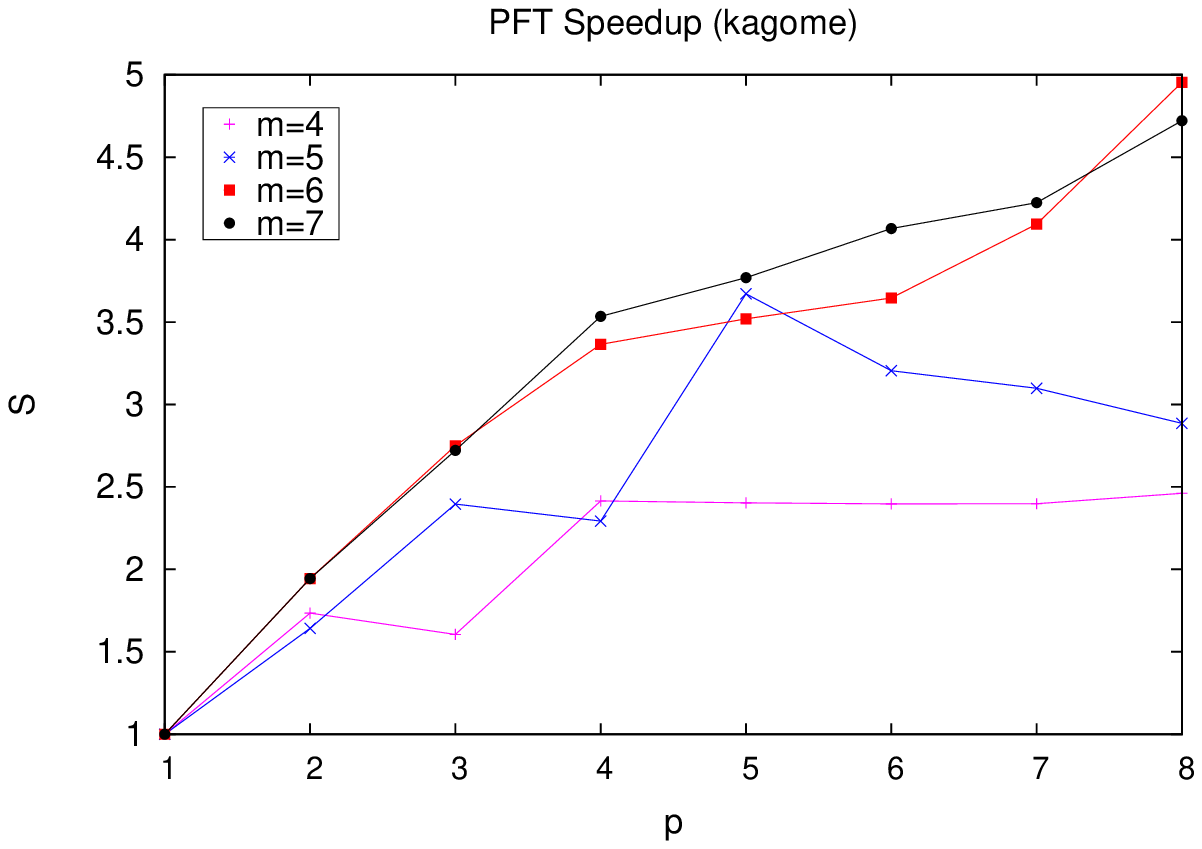}\\
\includegraphics[scale=0.59]{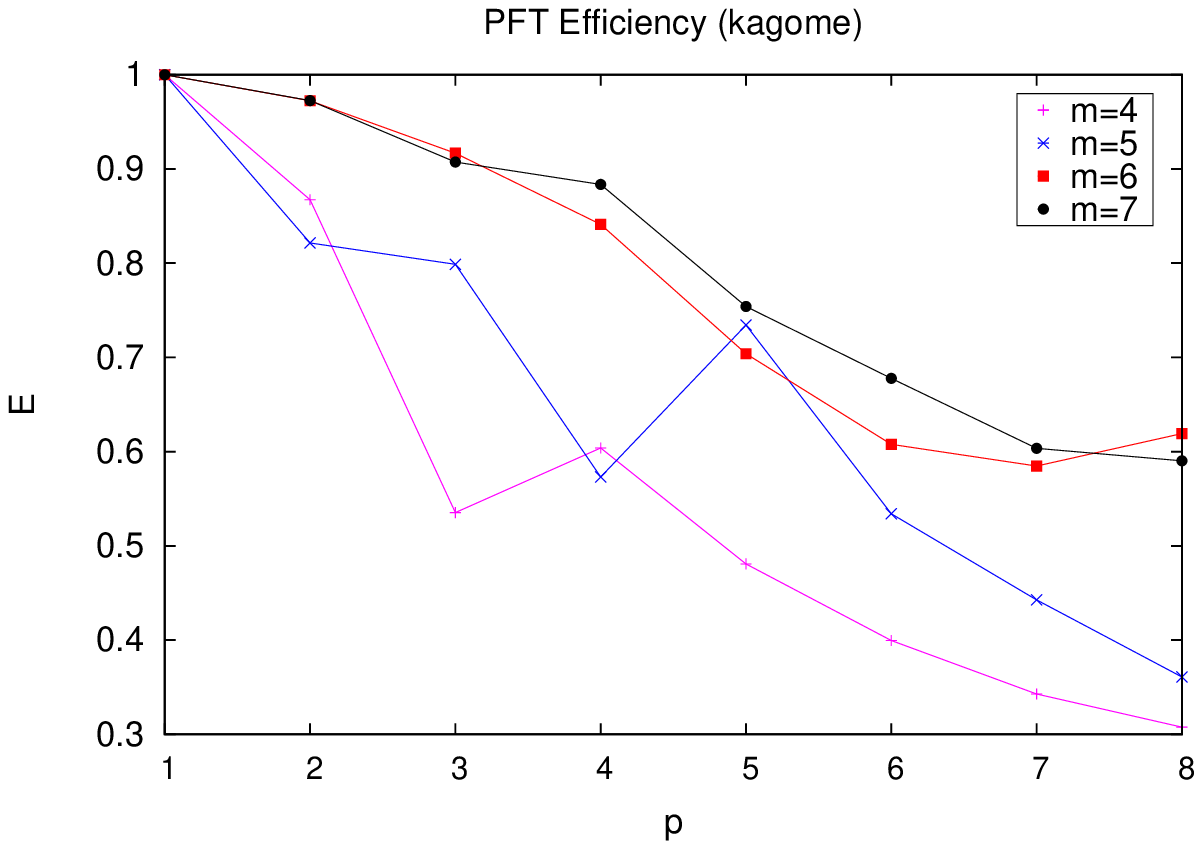} &
\includegraphics[scale=0.59]{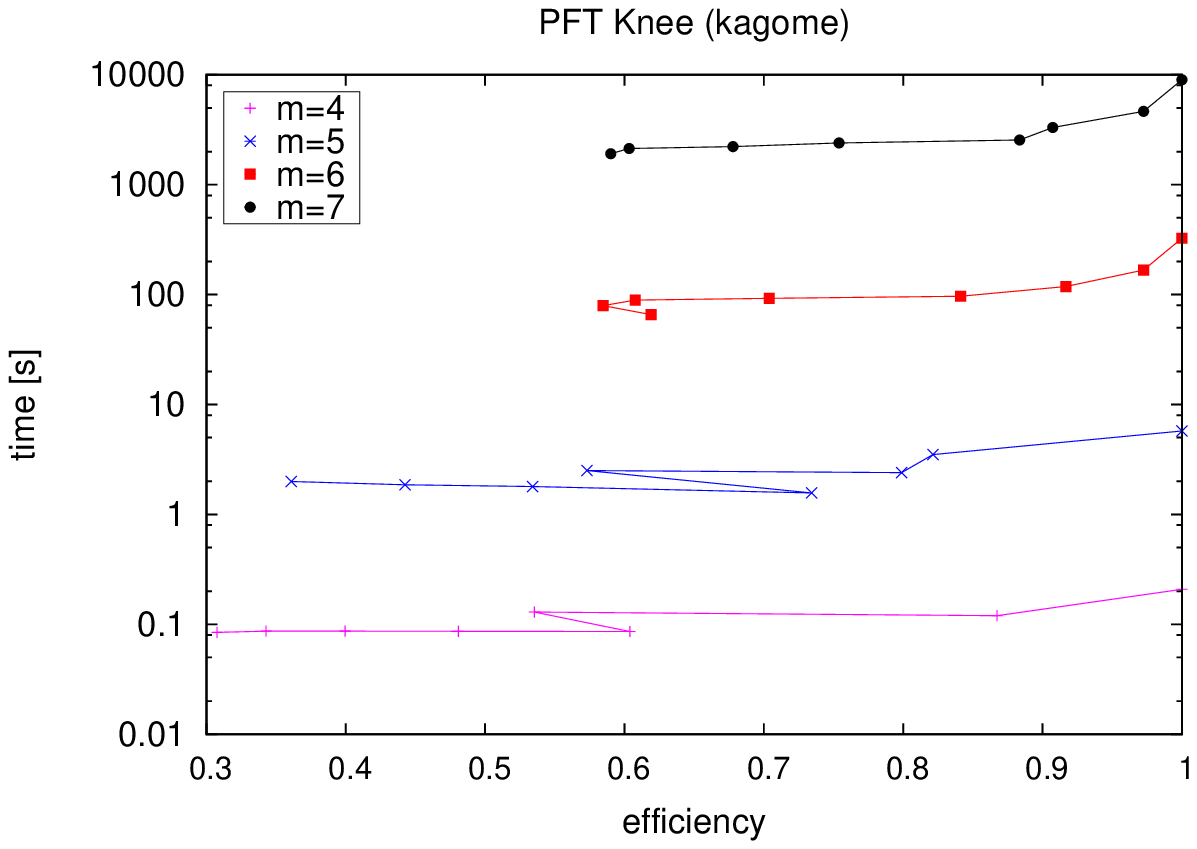}\\
\end{tabular}
\caption{Multi-core running time, speedup, efficiency and knee for the kagome strip test.}
\label{fig_kagome_performance_results}
\end{figure*}
From the results we have that the parallel performance is still 
scalable even for dense layers; the maximum speedup is over $4.7$ for $p=8$ on the largest problems. 
When $m>5$, the efficiency of the parallel implementation is approximately over $60\%$ for all 
values of $p$. In this test the knee is harder to identify, however for the largest problems one can see a 
small curve that suggests $p=4$ which is in fact 90\% 
efficient when solving large problems.

\subsection{Cluster results} 
The cluster used for the tests has a total four nodes; each one with 32GB RAM and two quad-core processors Xeon 5500 2.26 GHz. The full systems offers a 
total of 32 processing cores and 128GB RAM. The network is Ethernet gigabit centralized and the implementation of MPI is \textit{openMPI}.  
\subsubsection{Square results}
For the test of the square strip lattice in the cluster environment, we tested 9 different strip widths in the range $m \in [2,10]$. 
For each width, we measure 32 average execution times, one for each value of $p \in [1,32]$. This process is repeated for both static and dynamic scheduling. 
The standard error for each average execution time is below 5\%. For the dynamic scheduler we have chosen a block size value of $B=1$. This value of $B$ produces the highest amount 
of communication between the worker processes and the scheduler, hence the most dynamic scenario. Advantage of axial symmetry has also been taken.

Figure \ref{fig_cluster_square_performance_results} shows the performance measures of the running time, speedup, efficiency and the \textit{knee} \cite{DBLP:journals/tc/EagerZL89} for the 
cluster environment. Note that for each color (size), the solid and dashed lines represent static and dynamic scheduling, respectively.

\begin{figure*}[ht!]
\centering
\begin{tabular}{cc}
\includegraphics[scale=0.59]{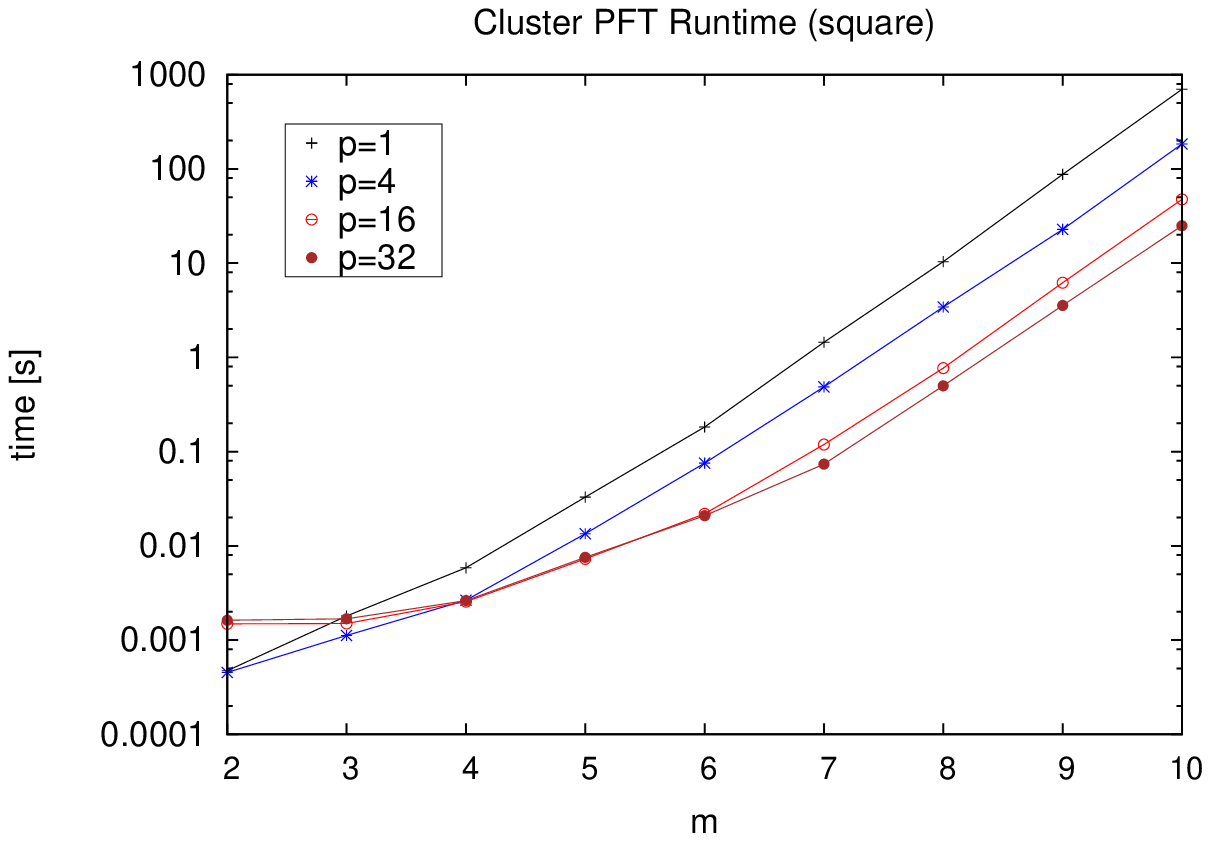} &
\includegraphics[scale=0.59]{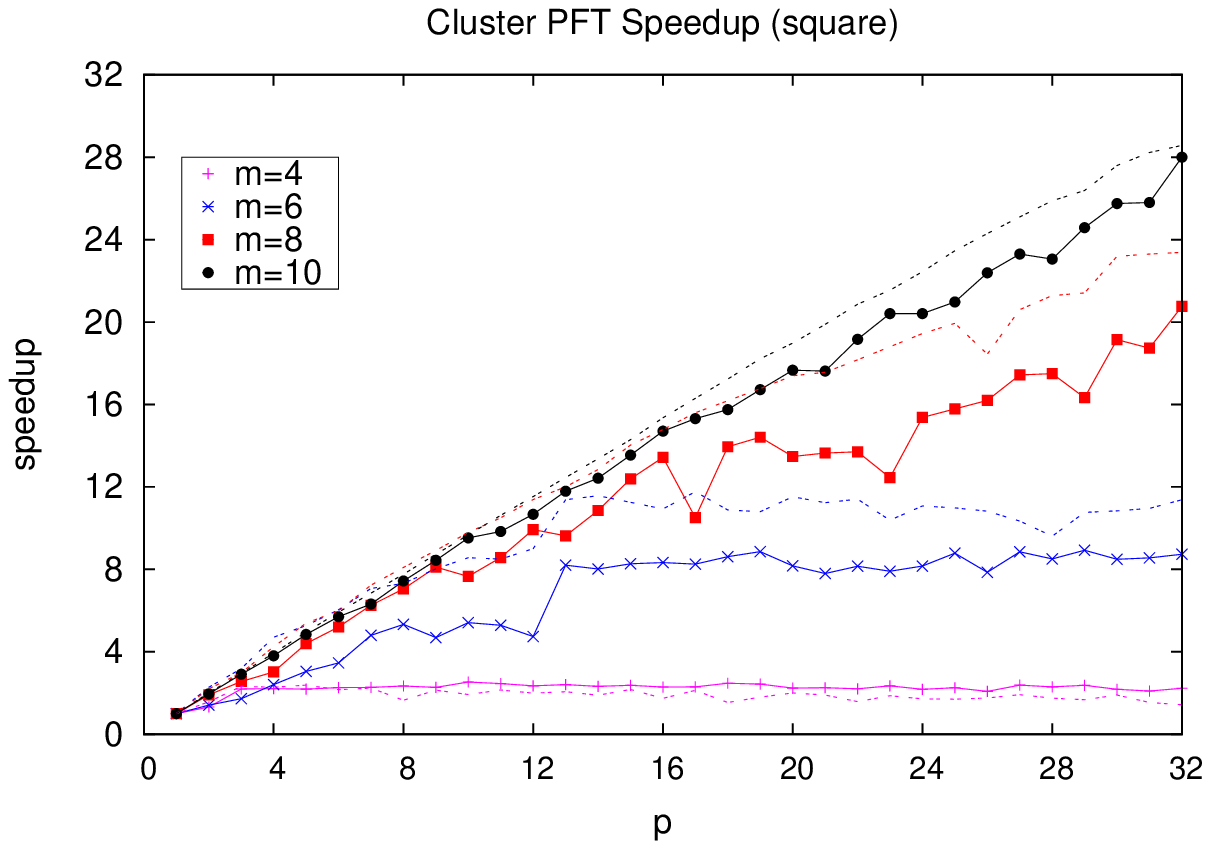}\\
\includegraphics[scale=0.59]{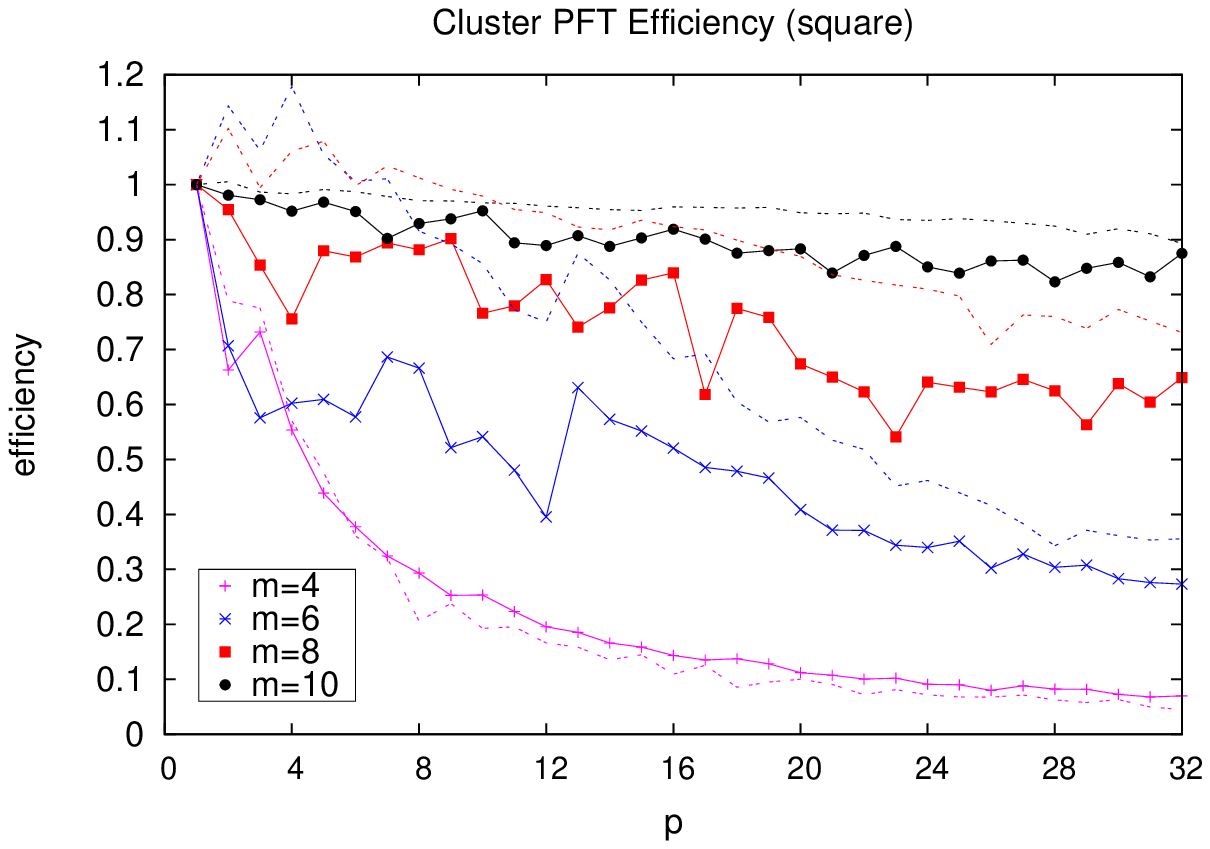} &
\includegraphics[scale=0.59]{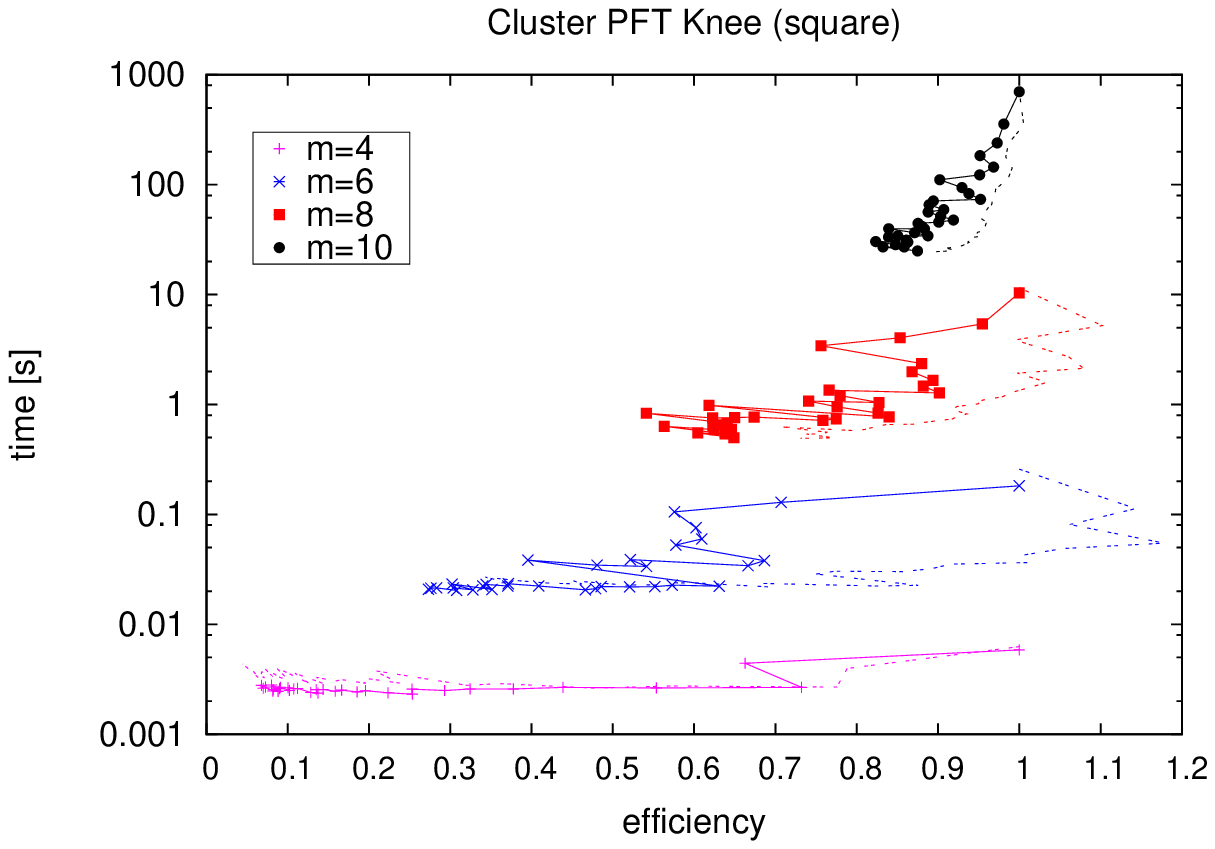}\\
\end{tabular}
\caption{Cluster running time, speedup, efficiency and the knee for the square strip test.}
\label{fig_cluster_square_performance_results}
\end{figure*}

From the results we observe that the reduction of the running time becomes effective starting from problems of size $m \ge 6$. Speedup has an overall linear behavior for the full range $p \in [1,32]$ which tells good scalability. 
Interestingly, near $p=4$ there is a region of \textit{super-linear speedup} \cite{wilkinson1999parallel} that occurs only for sizes $m=6,8$. For $p > 10$, super-linear speedup vanishes for all problem sizes.
In the cluster environment, the behavior between static (solid lines) and dynamic scheduling (dashed lines) is notorious; the former behaves irregularly producing several \textit{performance valleys}, 
while the latter behaves regularly, gives higher performance and produces close to zero \textit{performance valleys}. The maximum speedup achieved is approximately $28X$ for $p=32$, 
being superior in the dynamic case by a small margin. The efficiency of the parallel algorithm stays above $90\%$ for the largest case of $m=10$. Again, dynamic scheduler proves to 
be much more efficient than the static one when $m>6$, and overall the algorithm is over $70\%$ efficient for large enough problems, that is $m \ge 8$. 
The knee suggests that $p \in [8,10]$ gives the best balance of running time and efficiency whenever $m \ge 8$.

\subsubsection{Kagome results}
For the test of the kagome strip lattice in the cluster environment, we tested 5 different strip widths in the range $m \in [3,7]$. 
For each width, we measure 32 average execution times, one for each value of $p \in [1,32]$. This process is repeated for both static and dynamic scheduling. 
The standard error for each average execution time is below 5\%. For the dynamic scheduler we have chosen a block size value of $B=1$, same as in the square cluster test. 

Figure \ref{fig_cluster_square_performance_results} shows the performance measures of running time, speedup, efficiency and the \textit{knee} \cite{DBLP:journals/tc/EagerZL89} 
for the cluster environment. Note that for each color (size), the solid and dashed lines represent static and dynamic scheduling, respectively.
\begin{figure*}[ht!]
\centering
\begin{tabular}{cc}
\includegraphics[scale=0.59]{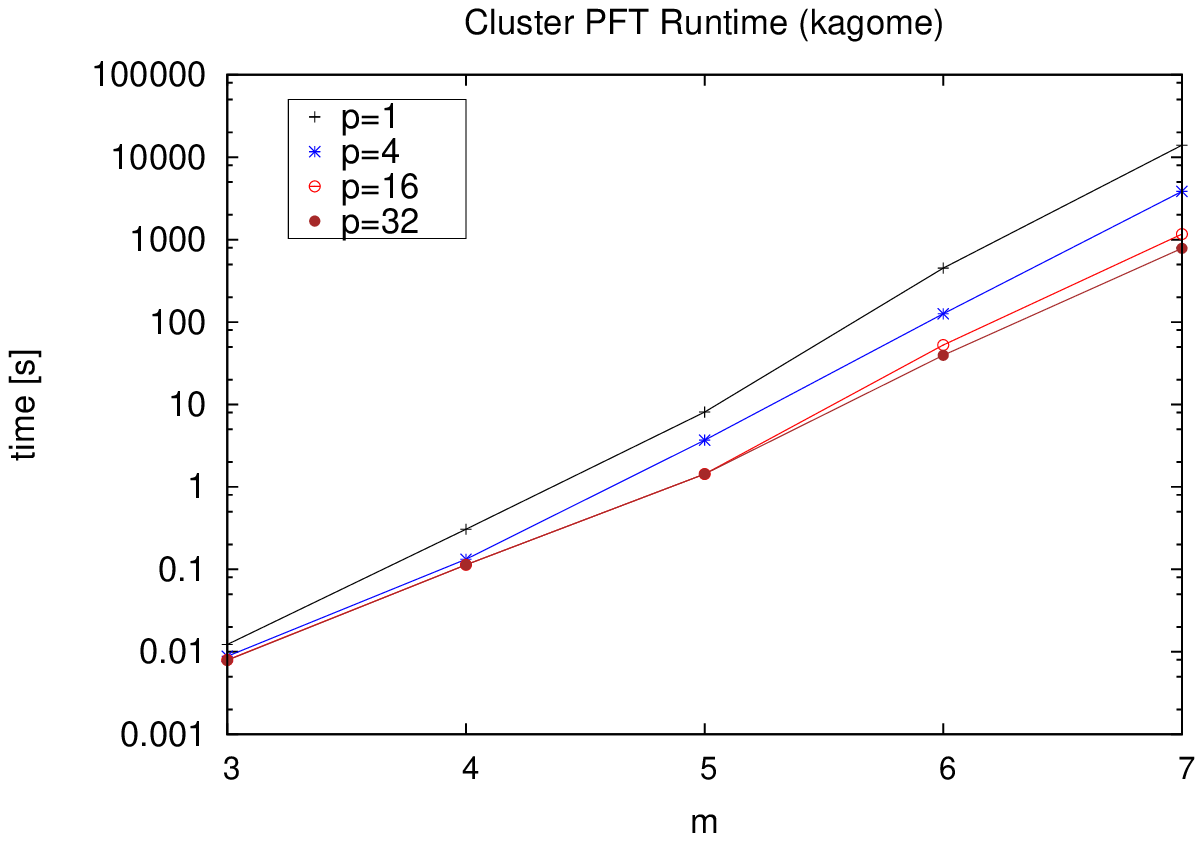} &
\includegraphics[scale=0.59]{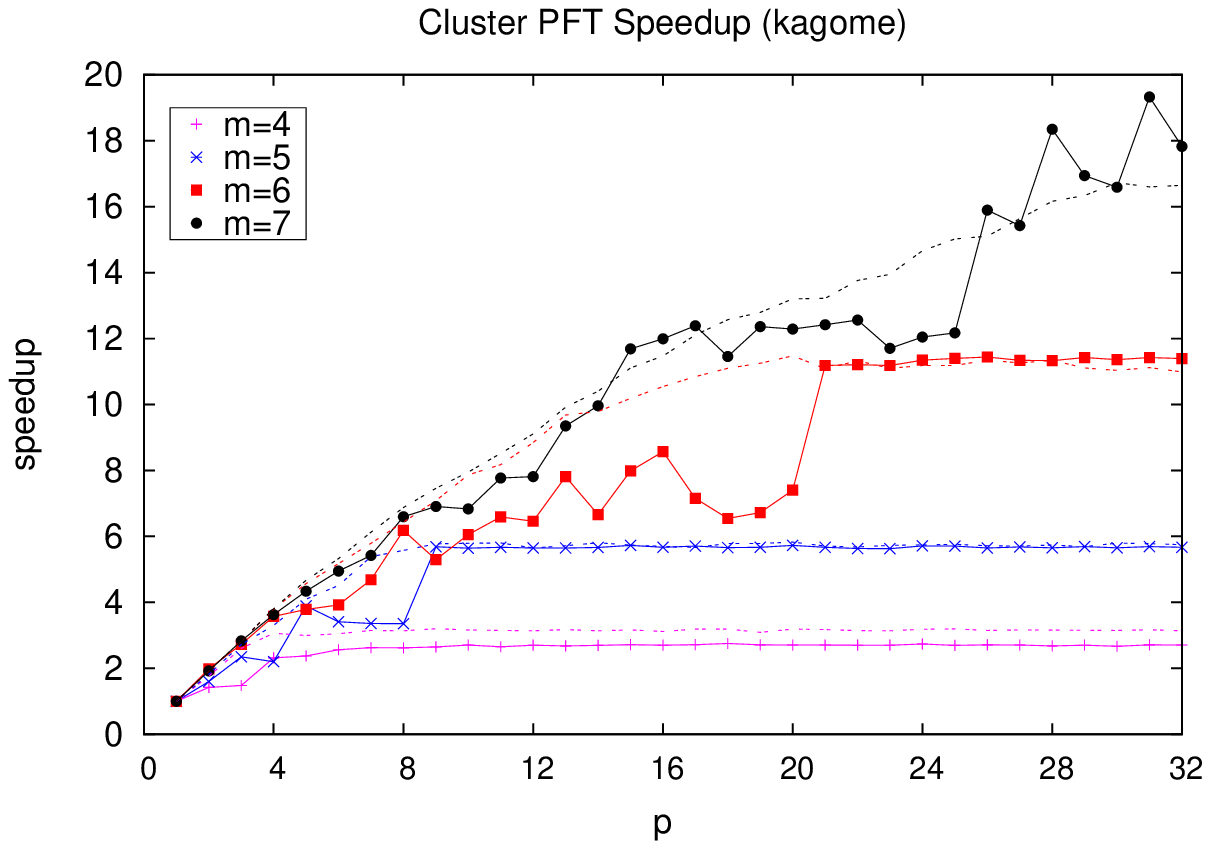}\\
\includegraphics[scale=0.59]{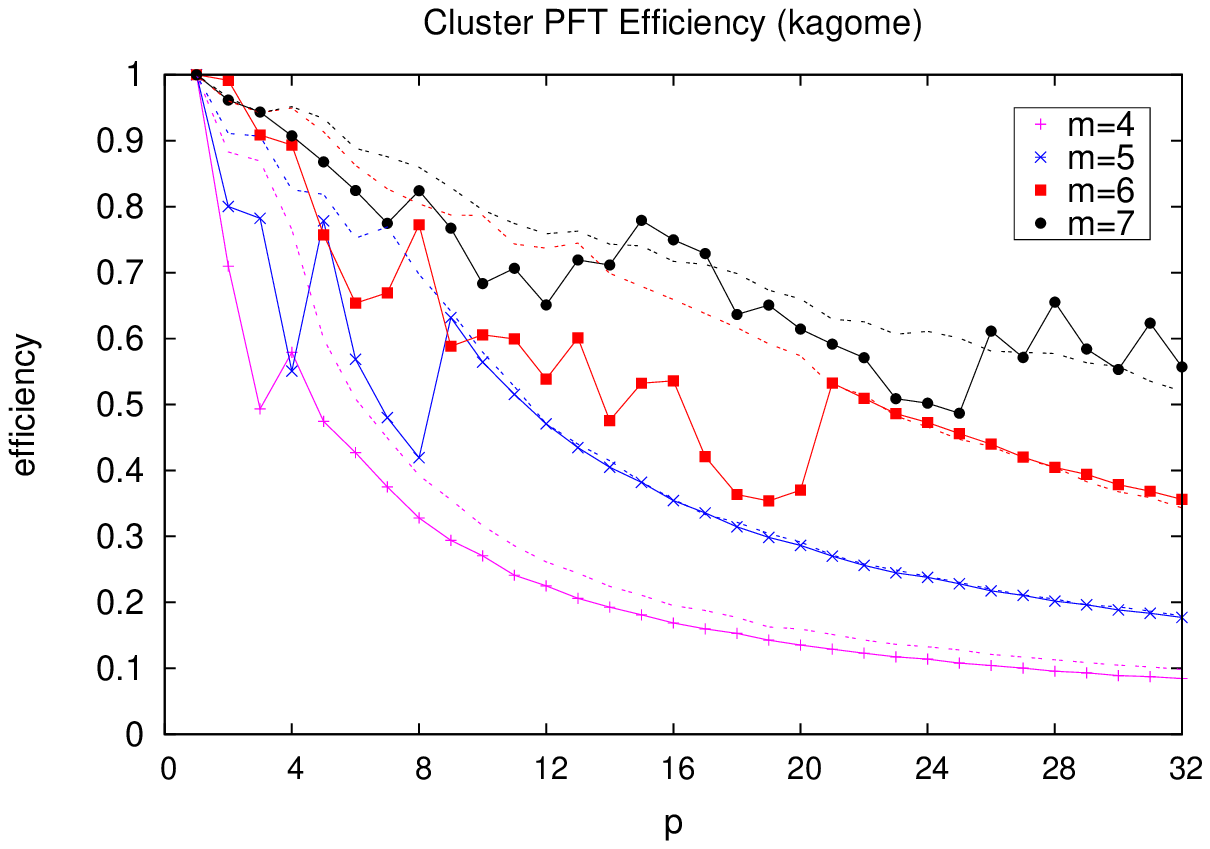} &
\includegraphics[scale=0.59]{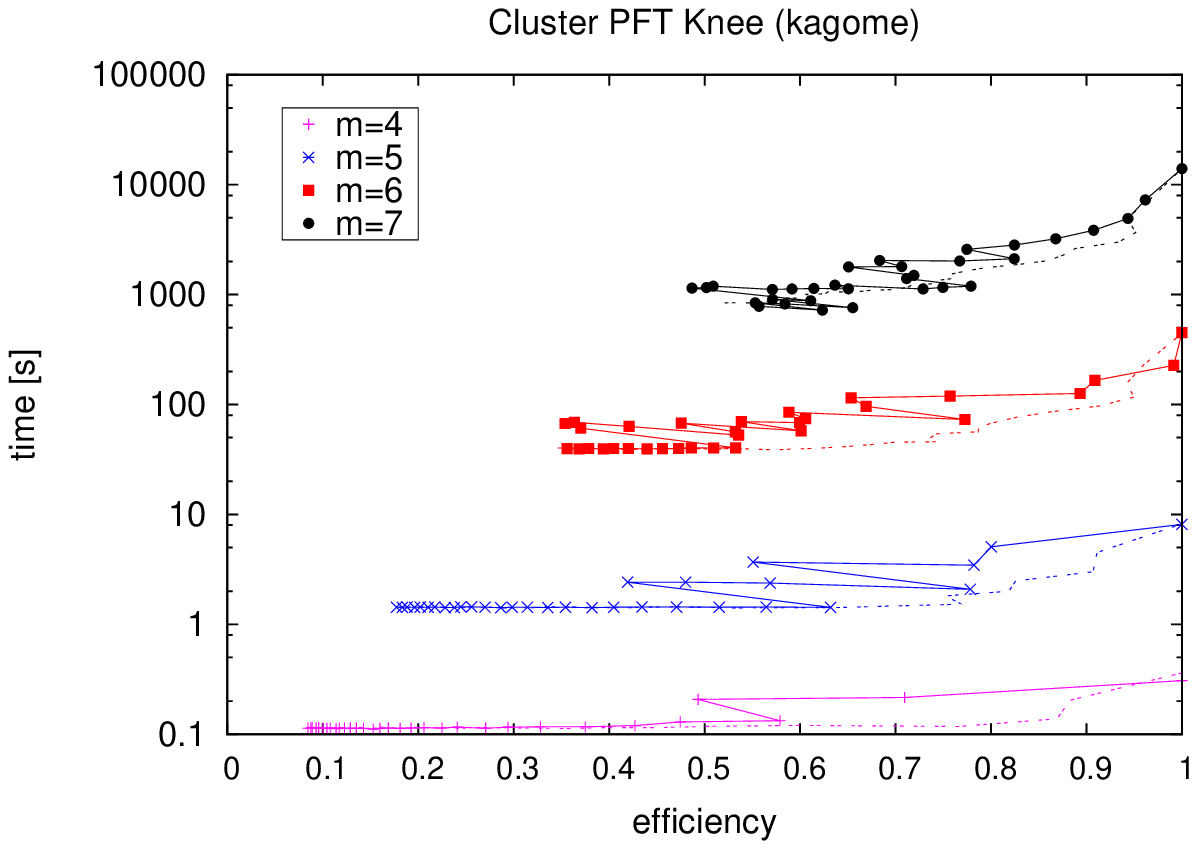}\\
\end{tabular}
\caption{Cluster running time, speedup, efficiency and knee for the kagome strip test.}
\label{fig_cluster_kagome_performance_results}
\end{figure*}

The results show that the reduction of the running time becomes effective in a cluster as long as $m \ge 6$. In this case, speedup is closer to a logarithmic curve rather than a 
linear one. It is interesting to note that speedup gets stuck at specific values for sizes $m=4,5,6$. The reason why is because 
the size of the configuration space is not large enough for cluster execution; $\Delta_m \le 32$ for $m = 4,5,6$. In fact, the values of $p$ where speedup starts to get stuck
actually match the values found for $\Delta_4, \Delta_5, \Delta_6$ in Table \ref{table_nonsym_sym_growth}. This phenomenon is totally normal in cluster or supercomputer 
environments, where the amount of work needed to reach full system occupancy is not always provided by the problem input. In order for speedup to take off, the configuration space 
must be equal or greater than the amount of processors available in the system.

There is a notorious difference in performance between static and dynamic scheduling. With dynamic scheduling, the \textit{performance valleys} are practically non-existent, giving a much more stable parallel 
performance for the full range of $p$. Efficiency is not as good as in the square test; the largest problem is solved with an efficiency over $55\%$, while the others reach below $50\%$ 
at some point of $p$. Dynamic scheduling proves to be in average more efficient than static scheduling, by-passing the \textit{performance valleys}. The Knee curve suggests a value 
$p \approx 8$ for a good balance between running time and efficiency.

\subsection{Impact of DC on algorithm performance}
We observed from the results that the running time of PFT applied to the kagome strip is slower 
than in the square strip. 
DC may cost too much in layers with a dense number 
of edges if optimizations do not occur too frequently. 
For the square lattice layer, we can write the DC worst case cost as
$O(2^{2m}) - O(opt) = O(4^m - opt)$ which is one of the fastest 
cases we can find, and optimizations, namely $O(opt)$, appear without too much effort. 
If we multiply this cost by the configuration space we have that the upper bound for the time to 
compute the transfer matrix of the square strip is 
$O(3^m \times (4^m - opt)) = O(12^m - 3^m \cdot opt)$, 
which is a notorious improvement with respect to the $O(16^m)$ bound 
with the standard Catalan technique, even if no DC optimizations occur.
Now for the kagome we can write the DC worst case cost as $O(2^{6m}) - O(opt) = O(64^m - opt)$ 
which would cost $O(3^m \times (64^m - opt)) = O(192^m - 3^m \cdot opt)$ in time 
when computing the  matrix. 
For dense layers the performance depends on how good the optimizations 
are and how frequently one can make them appear for a specific strip type. In our case the 
optimizations for kagome did not occur as frequent as in the square case because we programmed 
the heuristics in a very general way, nevertheless the method still managed to perform 
at least two times faster than the Catalan approach. 
It should be possible to make DC become more aware of the kagome structure and make it to generate 
the maximum number of optimization opportunities, as mentioned in the work of 
Haggard et. al. \cite{haggard_computing_tutte_polynomials}.

\subsection{Performance on wider strips}
We ran the PFT method to compute general $(q,v)$ transfer matrices 
on square strips at $m=\{11,12,13\}$ and kagome strips at $m=\{8,9\}$, using free boundary 
conditions and all the $32$ processors we had available. 
For the square strip, the computation of the TM took $\sim 5.5$ minutes for width $m=11$, 
$\sim 46$ minutes for width $m=12$ and $\sim 6.7$ hours for width $m=13$.
For the kagome strip, the computation of the TM took between $11 \sim 12$ hours at width $m=8$ and 
$\sim 3$ months at width $m=9$. 
These results were not included in the performance plots because it would have required 
excessive amount of time to benchmark for all values of $p$, specially for $p = 1$ where the 
computation is sequential. 
For the kagome strip we consider that we have reached the limit 
of tractability and wider kagome strips would become intractable\footnote{We consider that a problem becomes 
intractable when the time it takes to be solved is in the order of years for a given computer. 
It is possible that a faster computer can handle the problem, 
making it tractable.} with our hardware resources. 
For the square strip, we believe it is still possible to go further with our hardware resources, 
possibly up to $m=14$ or in the best scenario $m=15$ before reaching intractability. Moreover, 
if cylindrical boundary conditions are used, then it should be possible to go further beyond 
by using the symmetry of the dihedral group.

An important aspect of having a parallel solution is that if 
enough processors are used, that is $p = \Delta_m$, then the time for computing the 
transfer matrix becomes proportional to the depth of the largest directed-acyclic graph (DAG) 
of computation, which would correspond to the time required to solve the deepest family. 
The DAG concept allows to know what to expect when having more processors (\textit{i.e.}, a supercomputer) and 
gives insights on the limits of computation regarding parallelism.
If we apply the DAG concept to our results, we have that the time needed to compute the TM 
for the square strip would have been less than $5$ seconds for $m=11$ using $p = 1142$ processors, 
less than $10$ seconds for $m=12$ using $p = 2947$ processors and less than $5$ minutes for $m=13$
using $p = 7889$ processors. 
Analogous for kagome; the time needed to compute the TM would have been 
between $2 \sim 3$ hours for $m=8$ using $p = 70$ processors and 
$\sim1$ week for $m=9$ using $p = 323$ processors. 
As we mentioned earlier, DC heuristics that are aware of the kagome structure should improve the 
performance further.

\subsection{Comparison with related work}
In this subsection we compare the \textit{Parallel Family Trees} (PFT) strategy against the \textit{Catalan Parallel Method} (CPM) \cite{DBLP:conf/hpcc/NavarroHC13} by using the 
following metrics: (1) running time (2) matrix evaluation time and (3) matrix space. Figure \ref{fig_compare_results} shows the results.
\begin{figure*}[ht!]
\centering
\begin{tabular}{cc}
\includegraphics[scale=0.59]{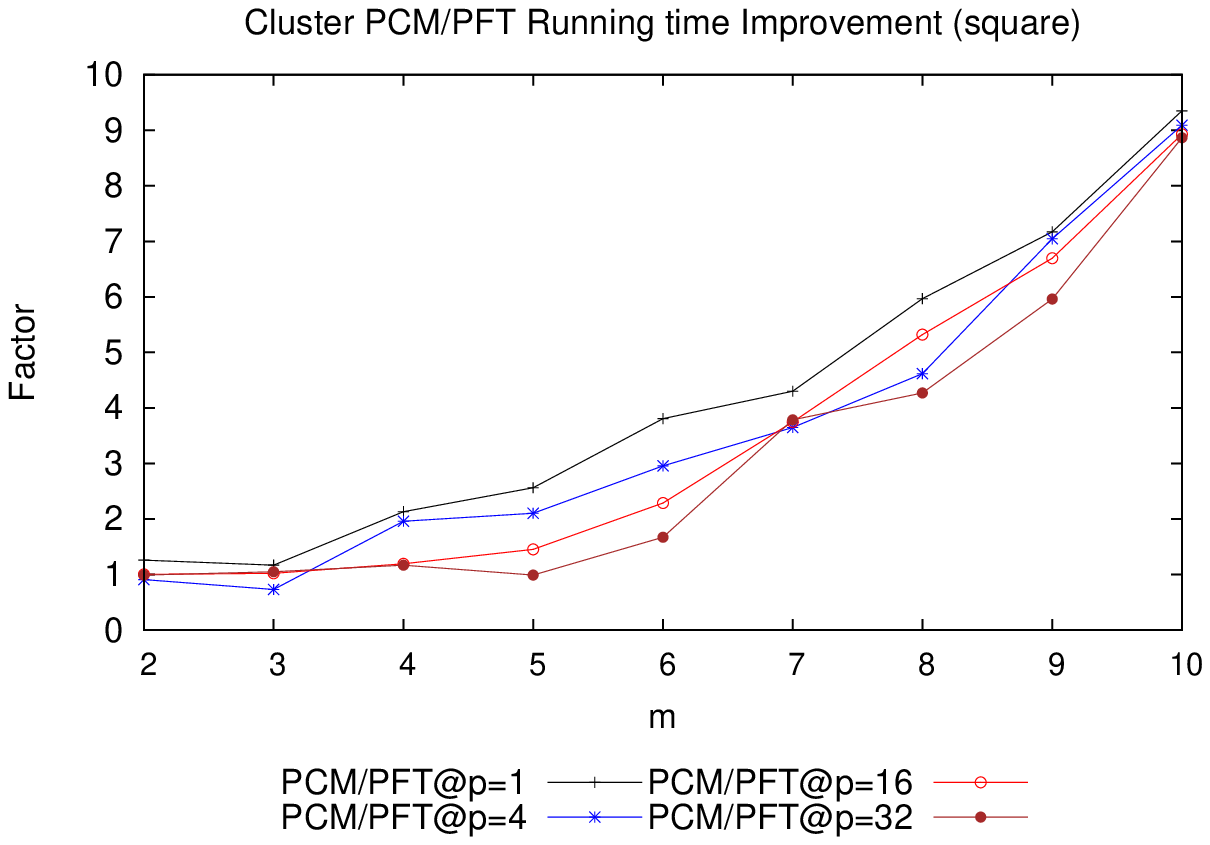} &
\includegraphics[scale=0.59]{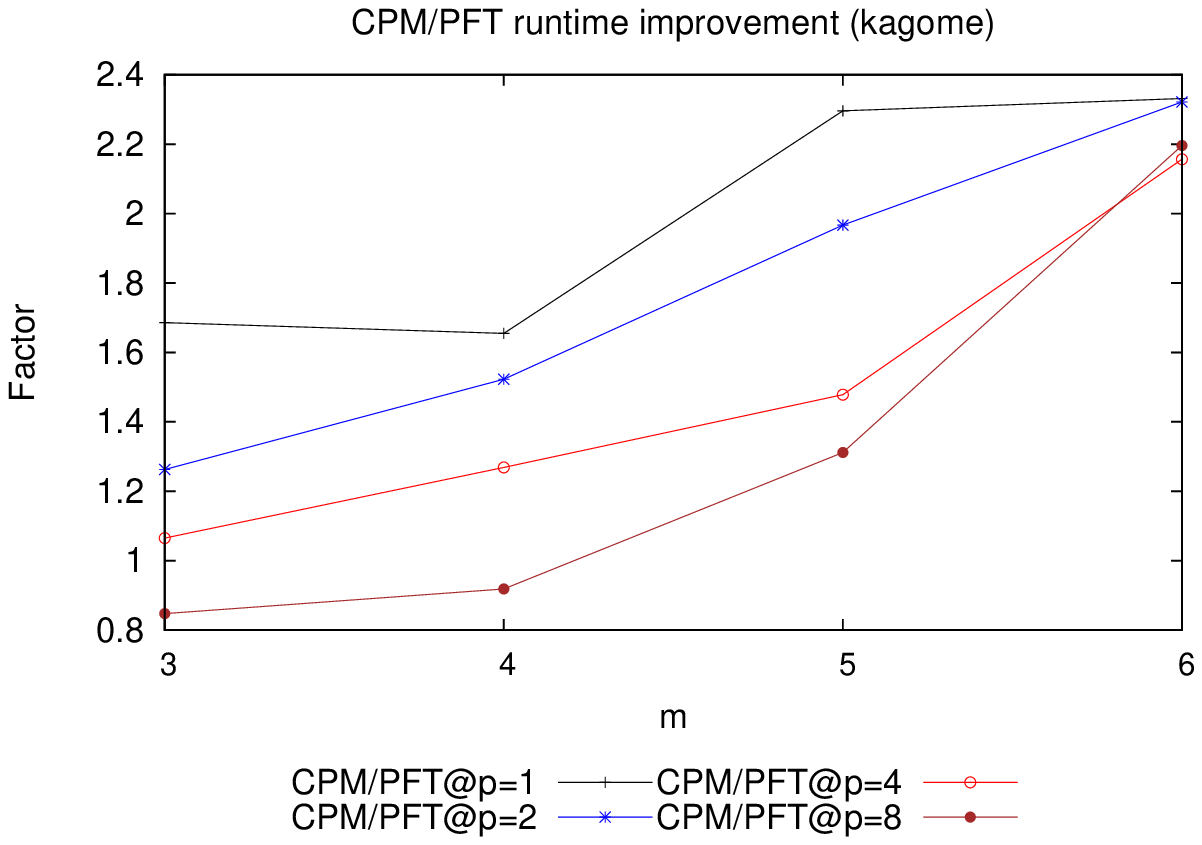}\\
\includegraphics[scale=0.59]{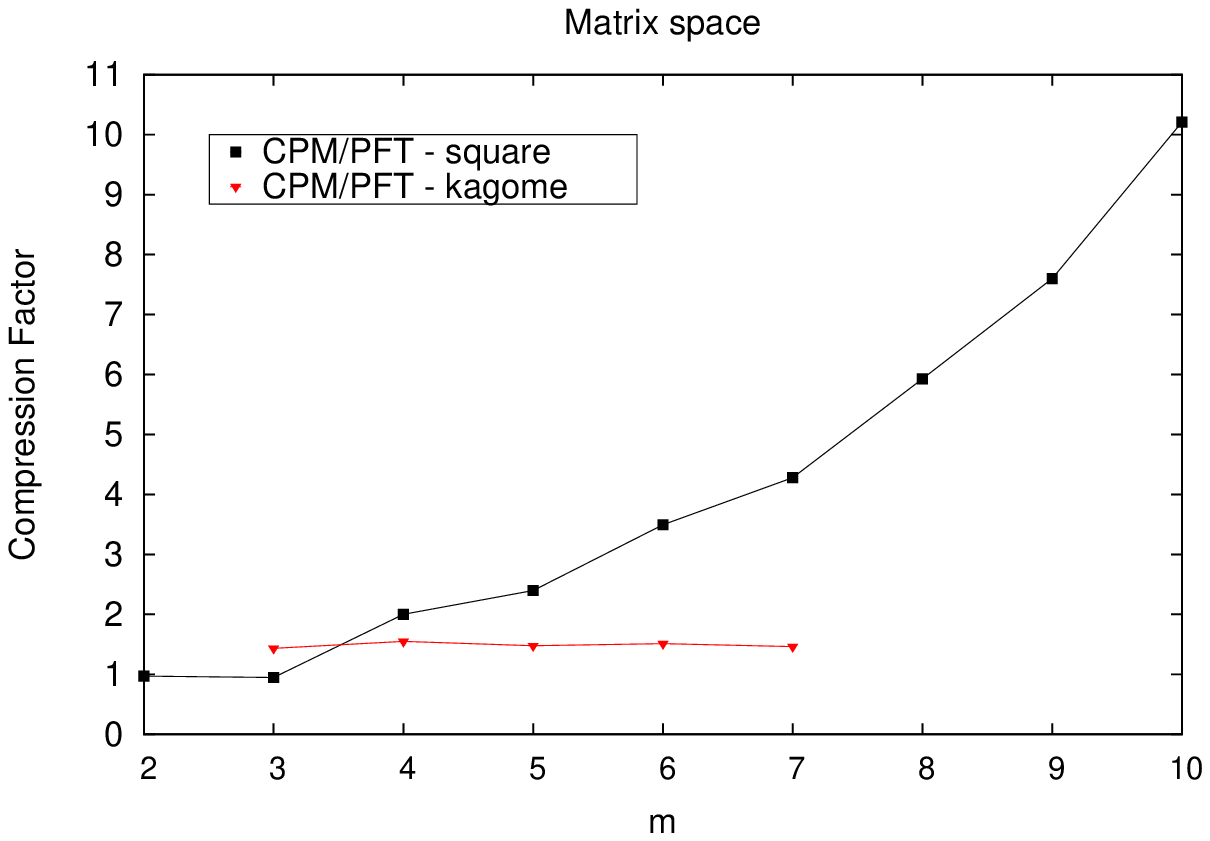} &
\includegraphics[scale=0.59]{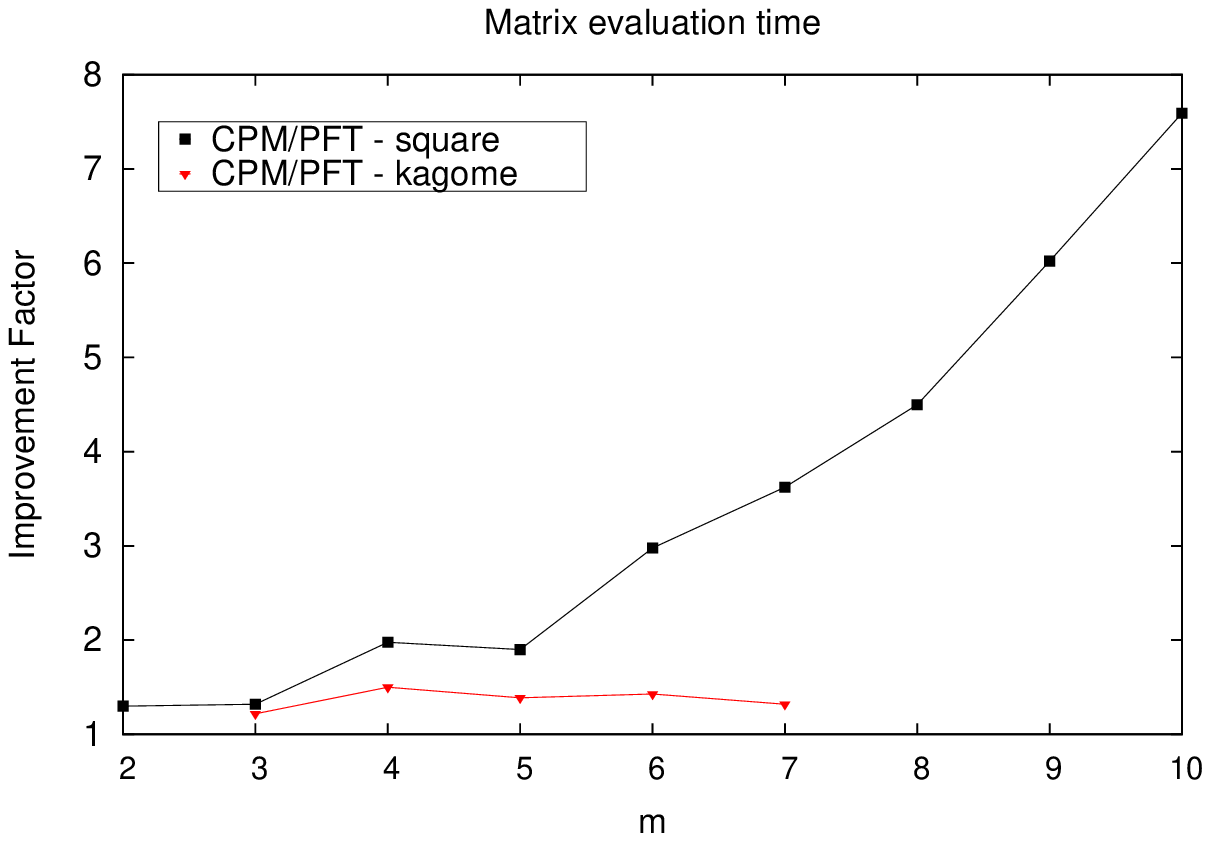}\\
\end{tabular}
\caption{Comparison between \textit{Parallel Family Trees (PFT)} and the \textit{Catalan Parallel Method (CPM)}.}
\label{fig_compare_results}
\end{figure*}

The first aspect to note from the running time results is that there is an non-linear improvement with respect to CPM that 
is independent of the amount of processors used. This improvement 
corresponds to the asymptotic reduction from $O(4^m)$ to $O(3^m)$ in configuration space. The 
improvement is less clear in the kagome strip test, but we expect that it should manifest when 
exploring larger sizes of $m$ or when using better heuristics for the DC optimizations. 
For the space metric, we observe that the size of the compressed matrices is indeed smaller 
than in the CPM case. Moreover, for the square strip the amount of compression 
increases non-linearly as we expected from the theoretical bound. For the kagome test, the 
compression factor stabilizes at approximately $1.5$. We believe that the reason 
why kagome compression stays fixed is because the kagome matrix is more sparse than in 
the square case, making the method to group zero-elements instead of large polynomials, 
reducing the compression factor from the maximum possible if the matrix was dense.
For the results of Matrix evaluation, we observe that 
evaluation and decompression on a PFT-matrix is faster than just evaluation on a CPM-matrix. 
The improvement seems to be a consequence of the compression factor achieved previously, 
since the behavior is similar. 

\subsection{Dynamic scheduler and block size}
The role of the block size under dynamic scheduling can be viewed as the amount of \textit{staticness} induced to the program. A value of $B=1$ means a fully dynamic scheduler, 
while a value of $B = \lceil n/p \rceil$ means a fully static scheduler. Given that the dynamic scheduler of our implementation communicates via 1-byte messages, it is safe to use 
$B$ as long as the network is sufficiently fast and dedicated to the cluster, like in our case. In a limited and shared network environment, one could consider exploring the range 
$1 < B < \lceil n/p \rceil$ until a good local minimum is found.

\subsection{Axial Symmetry}
When using axial symmetry, we observed an extra improvement in performance of up to $2X$ for the largest values 
of $m$. This improvement applies to both sequential and parallel execution. The size of the transfer matrix is improved under axial symmetry, in the best cases 
we achieved almost half the dimension of the original matrix, which in practice translates into up to $1/4$ of the space of the original non-symmetric matrix. 
Lattices as the kagome will only have certain values of $m$ where it is axial symmetric. In the other cases, one must perform a non-symmetric computation.

\section{Validation}
\label{sec_validation}
In this section we present some physical results we have computed for different widths of the 
square strip using free boundary conditions, as a way to validate the correctness of the 
\textit{parallel family trees} method by comparing the curves with the ones from related works. 

The first set of results are shown in Figure \ref{fig_fixv}. In the graphics we present 
the limiting curves on the complex $q$-plane for different values of the 
temperature-like parameter; $v = \{-1.0, -0.5, -0.1\}$, at different strip widths in the 
range $m \in [2,8]$.
The curves were obtained by using the \textit{direct-search approach} method which consists of 
scanning the complex domain in small discrete steps, and checking on each discrete location
the condition $|\lambda_1| = |\lambda_2|$ where 
$\lambda_1$ and $\lambda_2$ are the first and second dominant eigenvalues, respectively.  
If the condition is true, then the pair $(x, y)$ is a point of the curve, 
where $x$ and $y$ are the real and imaginary parts of $q$, respectively. Due to numerical 
precision limits, we allowed $\%1$ of numerical error for accepting 
the condition $|\lambda_1| = |\lambda_2|$. For the case of $v=0.5$ we allowed up to $\%4$ of error 
for drawing the limiting curve at size $m=8$.
\begin{figure*}[ht!]
\centering
\begin{tabular}{ccc}
\includegraphics[scale=0.775]{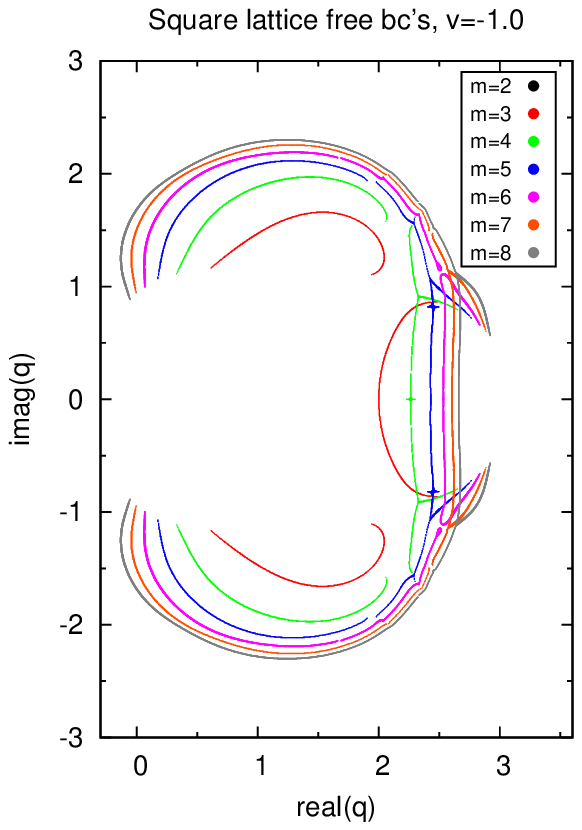} &
\includegraphics[scale=0.775]{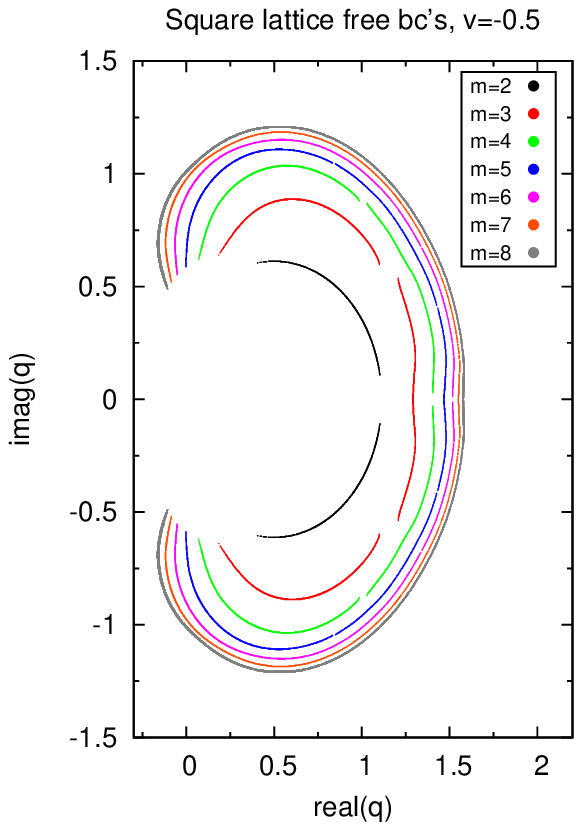} &
\includegraphics[scale=0.775]{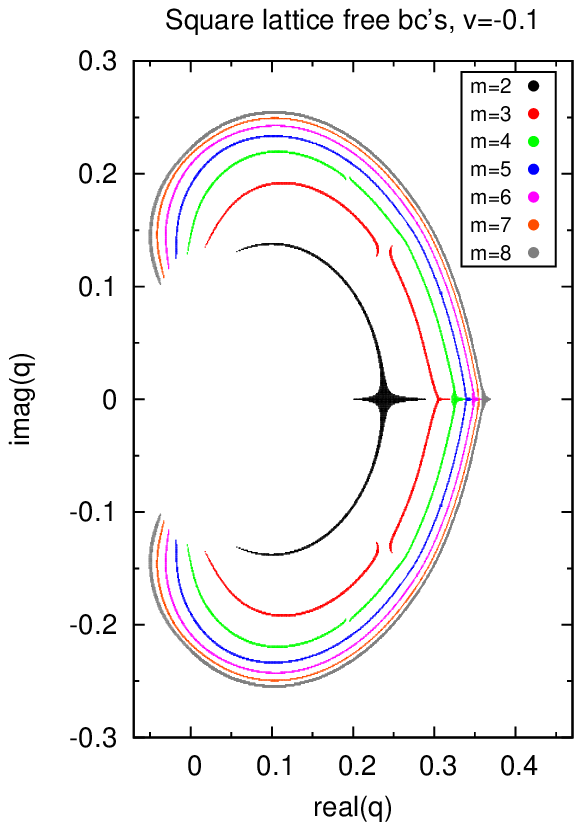}
\end{tabular}
\caption{Limiting curves on the complex $q-$plane for $v = \{-1.0, -0.5, -0.1\}$. In each graphic there are seven 
		 limiting curves with different colors, each one corresponding to a different strip width.}
\label{fig_fixv}
\end{figure*}
The curves for $v=-1$ agree with the ones presented by Salas \textit{et. al.} in Figure 21 of ref. \cite{salas_sokal_1}. 
The curves for $v=-0.5$ and $v=-0.1$, although grouped in a different way, 
agree with the result obtained by Chang \textit{et. al.} from Figures 2, 3, 4 of ref. \cite{salas2002}.
Limiting curves for $6 \le m \le 8$ did not appear in the cited work.

For the next set of physical results we are interested in fixing the $q$ parameter at values 
$q = \{2, 3, 4\}$ and compute the dimensionless reduced internal energy $E_r$ as well as the 
reduced function $C_H$ of the specific heat $C$, for different strip widths in the 
range $m \in [2,8]$. The dimensionless reduced internal energy is defined as
\begin{equation}
E_r = -\frac{E}{J} = (v+1)\frac{\partial f}{\partial v}
\end{equation}
where $f$ is the free energy density as defined in equation (\ref{eq_freeenergy_density}), 
$J$ the coupling constant which is $J > 0$ for the \textit{ferromagnetic} case ($0 < v < \infty$) 
and $J < 0$ for the \textit{antiferromagnetic} case ($-1 < v < 0$). The specific heat is defined as
\begin{equation}
C = \frac{\partial E}{\partial T} = k_B K^2 (v+1) \Bigg[\frac{\partial f}{\partial v} + (v+1)\frac{\partial^2 f}{\partial v^2}\Bigg]
\end{equation}
and $C_H$ uses the reduced form
\begin{equation}
C_H = \frac{C}{k_B K}
\end{equation}
The results are presented in Figure \ref{fig_energyspecheat}, where each row presents the results for a 
given $q$ value. The curves for $2 \le m \le 5$ agree with the ones presented 
by Chang \textit{et. al.} \cite{salas2002}. Results for $6 \le m \le 8$ did not appear in the cited work.
\begin{figure*}[ht!]
\centering
\begin{tabular}{cc}
\includegraphics[scale=0.75]{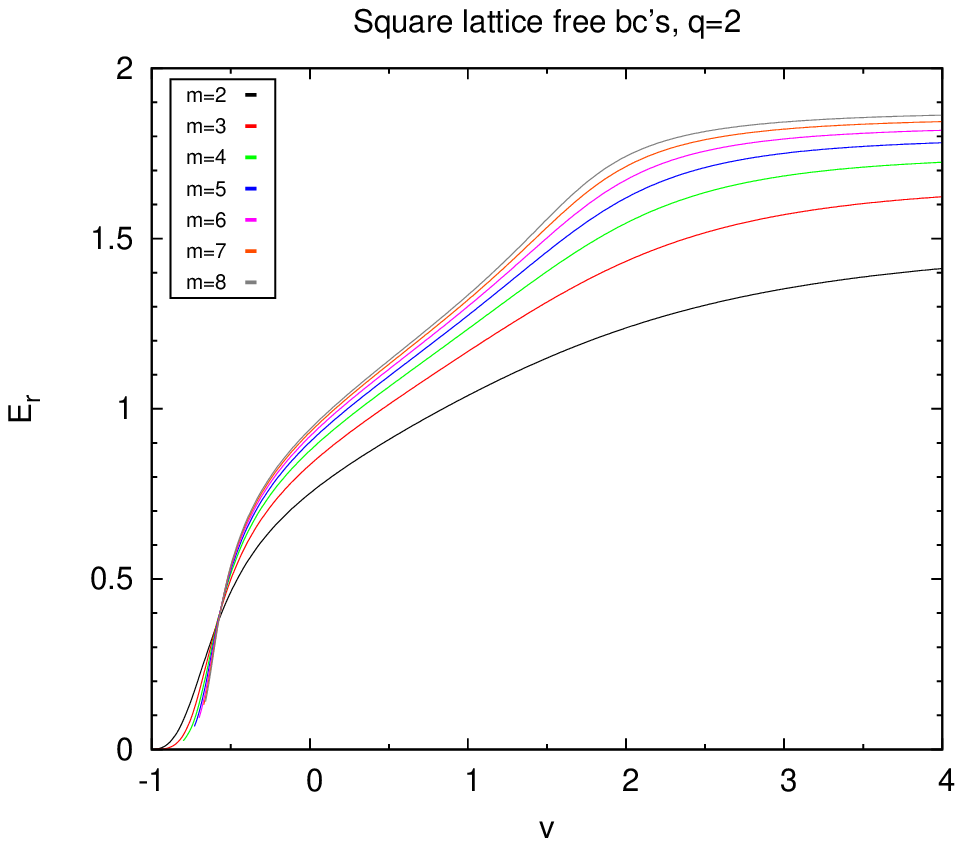} &
\includegraphics[scale=0.75]{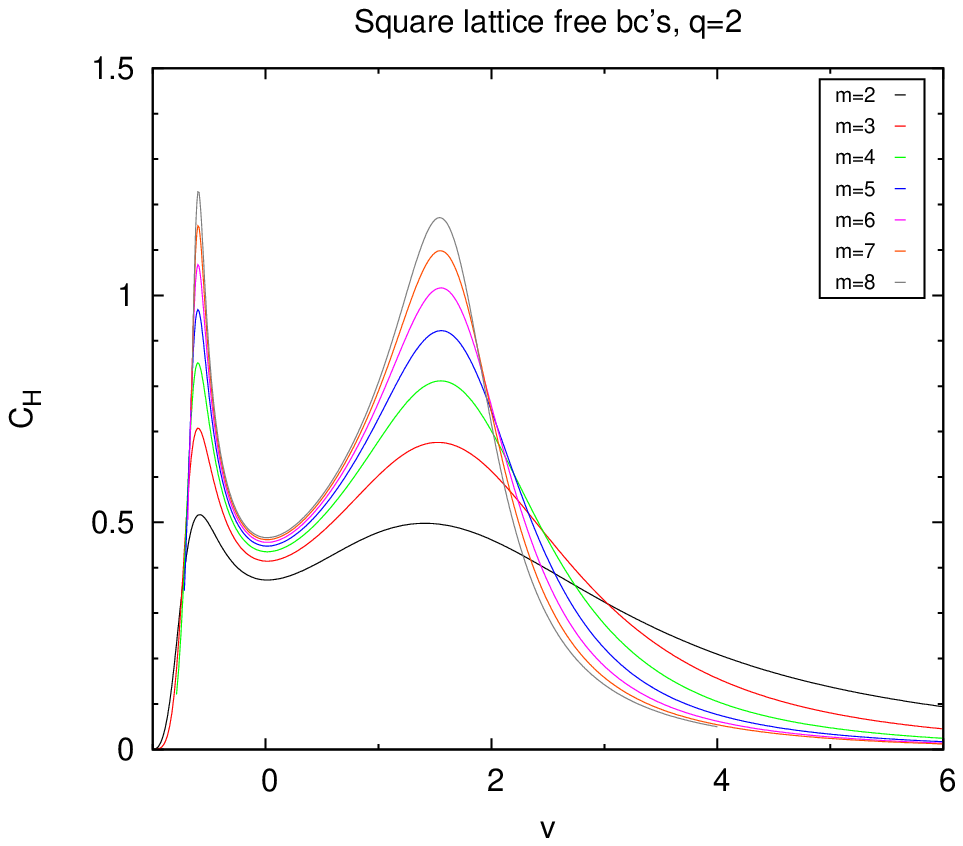}\\
\includegraphics[scale=0.75]{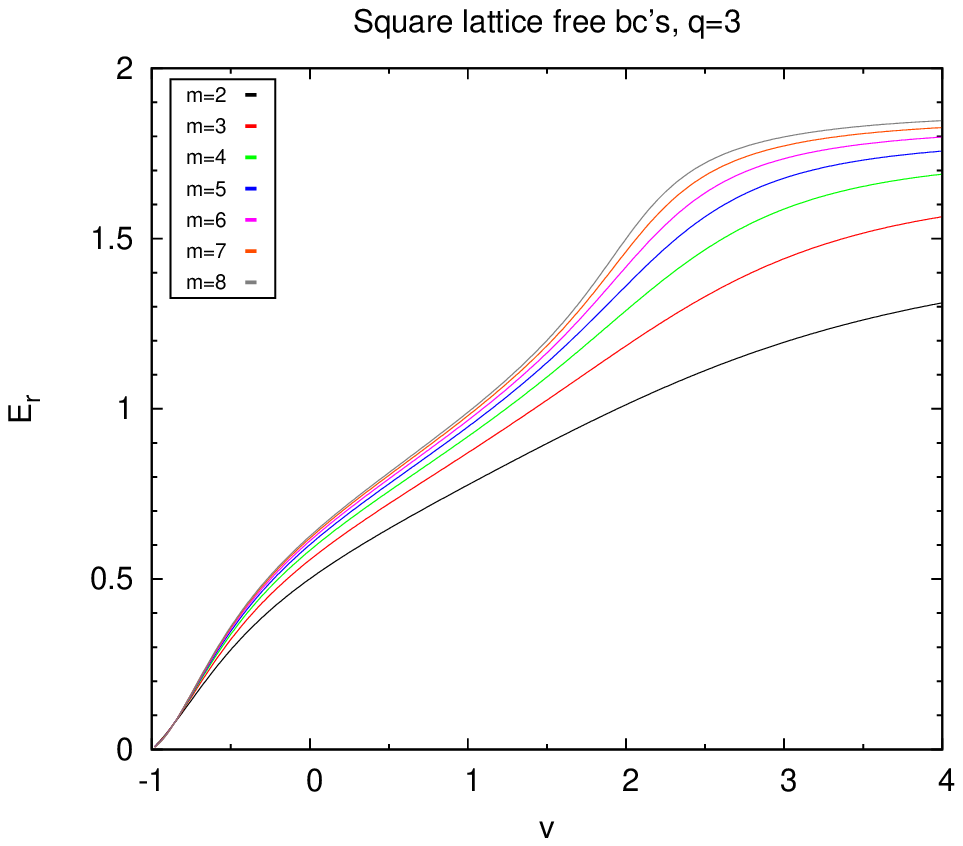} &
\includegraphics[scale=0.75]{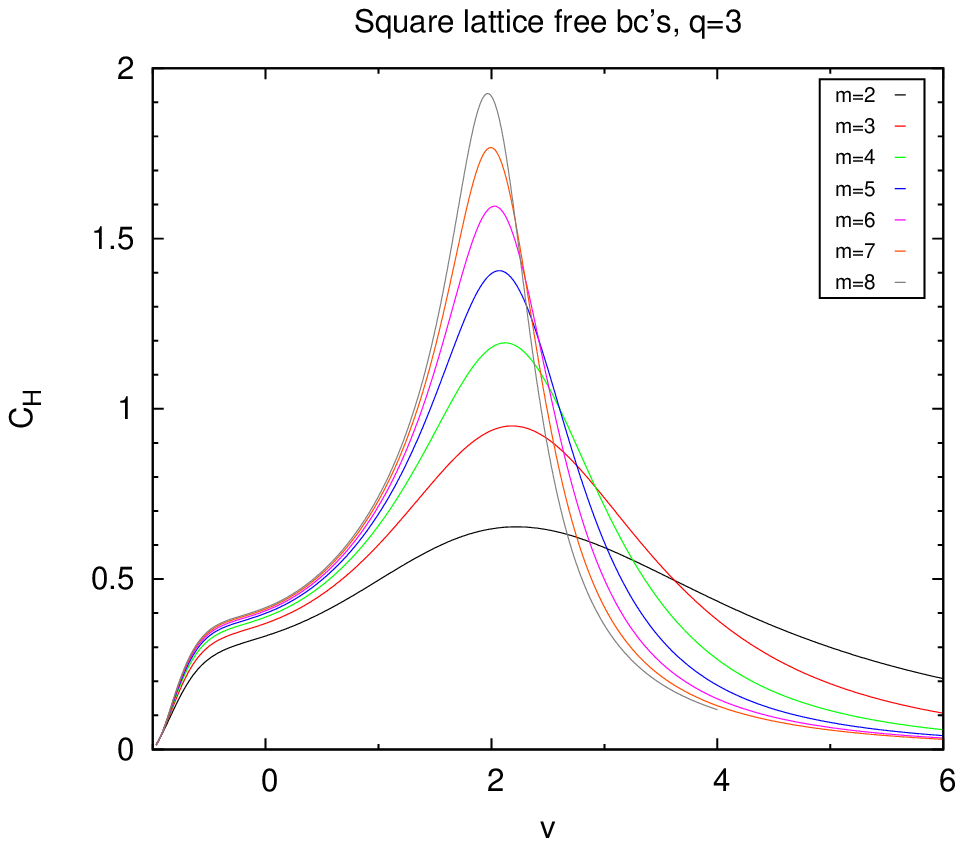}\\
\includegraphics[scale=0.75]{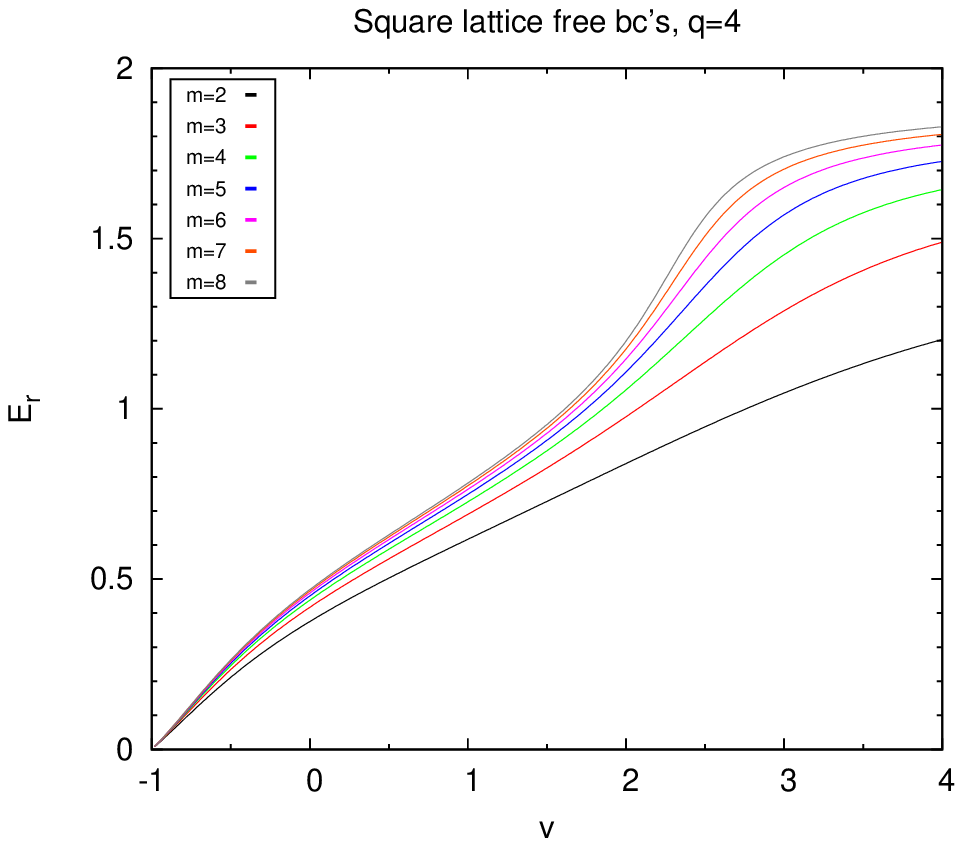} &
\includegraphics[scale=0.75]{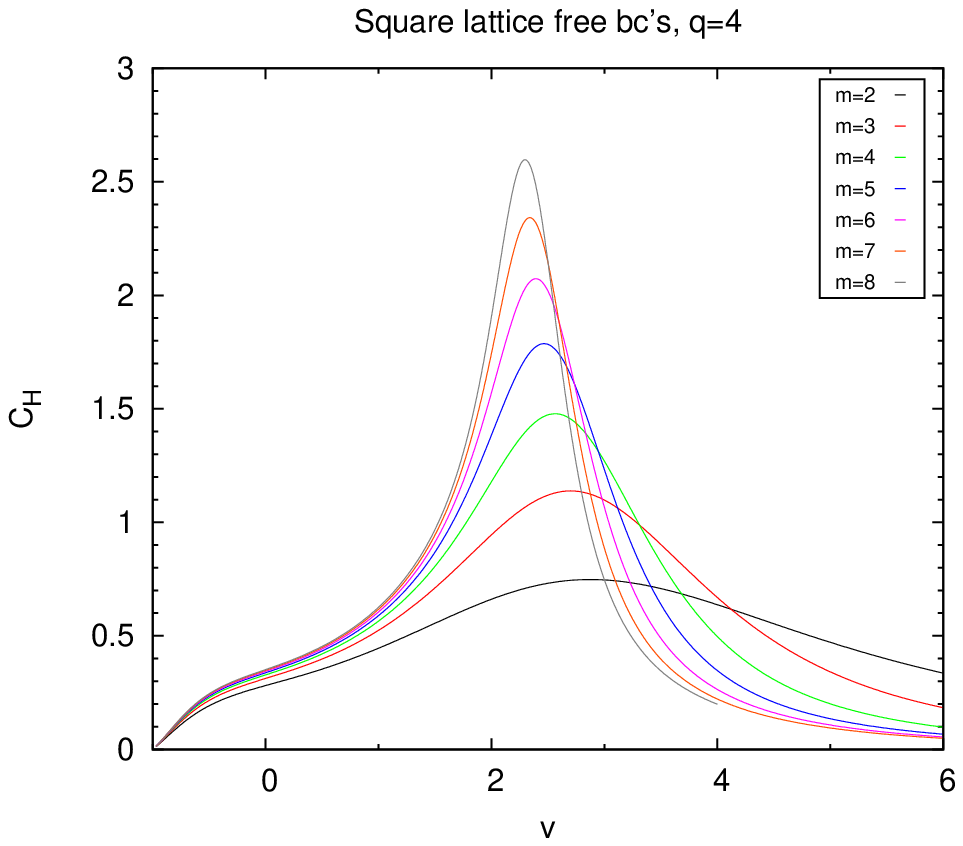}\\
\end{tabular}
\caption{Plots for reduced internal energy $E_r$ and reduced specific heat $C_H$ for 
		$q = \{2, 3, 4\}$.}
\label{fig_energyspecheat}
\end{figure*}

Although the computation of new physical curves 
for wider strips is indeed possible, it would require more time with our resources, 
or a much larger cluster than ours for faster results. 
Nevertheless, our present results already show that with the \textit{PFT} strategy known results 
are obtained faster than with CPM.
We would like to remind the reader that the focus of this work is on the algorithmic 
improvements and the possibilities to compute the general $(q,v)$ transfer matrix for strips, 
using a configuration space that is asymptotically $O(3^m)$ . 
\section{Discussion}
\label{seq_conclusions}
We have presented a parallel strategy for computing the general $(q,v)$ transfer matrix of strip lattices in the Potts model. Our main result is the asymptotic reduction 
of the configuration space, from $O(4^m)$ to $O(3^m)$, by re-organizing the problem domain as 
\textit{parallel family trees (PFT)}. Using this strategy, 
the transfer matrix can now be computed by just processing the root configurations, 
which are $O(3^m)$ in number. 
Computation of the family trees can be performed completely in parallel because family 
trees are independent from each other, 
and the configuration space is generated \textit{a priori}, removing any potential 
time-dependence. We have compared the experimental results of PFT and indeed 
it runs exponentially 
faster than the \textit{Catalan Parallel Method (CPM)} \cite{DBLP:conf/hpcc/NavarroHC13}, both in 
sequential and parallel execution. 

The resulting matrix of PFT is a compressed structure based 
on systems of linear equations. Numerical evaluation on the 
matrix, including decompression time, is actually faster than numerical evaluation using the CPM 
method, by a factor that is proportional to the improvement we measured for running time. 
Therefore, it is not only faster to generate the matrix using PFT, but it is also faster 
to use it later for extracting the physical information.

Multi-core results have shown that PFT benefits from shared-memory parallelism, achieving a maximum of $5.7$X of speedup for the square strip test when using $p=8$ processors. 
At $p=4$, the efficiency of the implementation is still over $95\%$, which is worth mentioning. By plotting the \textit{knee} curve, we have managed to confirm that choosing $p=4$ 
is in fact a wise decision for a balance of speed and efficiency. In the Multi-core scenario, a dynamic scheduler did not produce a beneficial change in performance, therefore static scheduling still remains convenient.

For the cluster results, we achieved up to $28X$ of speedup using $p=32$ for the square strip tests, with an efficiency above $90\%$ for a strip of width $m=10$ (largest one). 
For the kagome strip test, efficiency stayed above $55\%$ for a strip of $m \ge 7$ and the maximum value of speedup reached was close to $20X$ when using $p=31$. 
A small \textit{super-linear speedup} region emerged near $p=4$ when solving square strips of sizes $m=6,8$, giving an efficiency of up to $120\%$. 
We believe that this is just a particular fortunate event, possibly produced by the reduction of cache misses, which is caused when partitioned data fits entirely in cache. 
In general, we do not expect \textit{super-linear} behavior since we are measuring \textit{fixed-size speedup} which is upper bounded as $S_p \le p$ \cite{gustafson1990fixed}.   
The knee curve suggests that $p \in [8,10]$ produces a good balance between speed and efficiency. An important result in cluster execution is that 
dynamic scheduling is mandatory in order to achieve a performance curve that will not fall into \textit{performance valleys}, as static scheduling did. On average, dynamic 
scheduling achieves considerable higher performance than static scheduling. 

One of the goals of this work was to present an algorithmic improvement that is implicitly parallel and scalable. 
For this, we introduced a preprocessing step that generates all possible \textit{root configurations} and \textit{Catalan configurations}, which are critical for processing the 
family trees in parallel. This step takes a small amount of time compared to the whole problem. 
Other technical improvements had been introduced, some of them being already known in the literature \cite{haggard_computing_tutte_polynomials}; 
(1) fast computation of serial and parallel paths of the graph, (2) exploiting axial symmetry, (3) a set of algebra rules for making consistent keys in all leaf nodes and 
(4) a hash table for accessing column values of the transfer matrix. In particular, when taking advantage of axial symmetry, the implementation achieved extra improvement of up to $2X$ in performance, 
using almost a quarter of the matrix space used in a non-symmetric computation.

In order to achieve a scalable parallel implementation, some small data structures were replicated among 
processors while some other data structures per processor were created within the corresponding worker process context, not in any master process. 
This allocation strategy results in faster cache performance and brings up the possibility to scale better under NUMA architectures.
It is not a problem to store the matrix fragmented into many files as long as the matrix is in its symbolic form. In practice, it is first necessary to evaluate the matrix 
on $q$ and $v$ before doing any further numerical analysis. Therefore, the fragmented parts can be evaluated at runtime as they become read. This evaluation can also be done in parallel.

The only technical restriction of the \textit{parallel family trees} strategy in order to work 
is that vertices of the left and right boundaries of the layer need to be connected sequentially. 
This restriction is not a problem, because any planar strip lattice can be rotated 
so that the restriction 
is satisfied. Additionally, PFT allows any graph structure along the vertical direction, 
that is, one can study strips where its $K_i$ layer is composed by a sequence of different 
tiles. 

In the kagome tests, the performance results were not as good as we expected, 
because the number of edges in the layer is much higher than in the square case, making DC to take a 
considerable amount of time for each configuration. 
We believe that the dependence of DC on the number of edges in the layer is a sensible aspect for 
the PFT algorithm, and an extrapolation of this situation would suggest that 
the largest Archimedean lattices could be much harder to the point of being intractable. 
However, it is important to consider that DC can significantly improve its performance 
if the heuristics are improved so that they choose the best sequence of edges 
based on the connectivity of the graph layer \cite{haggard_computing_tutte_polynomials}. 
These heuristics, combined with the linear-cost optimizations, can make the PFT method 
more resistant to the number of edges in the layer. Furthermore, if more processors are used to the 
point that $p = \Delta_m$, then the time for computing the TM will be much lower than in our case 
with $p=32$, and will correspond to the time taken to solve the deepest DAG of computation. 
For this reason, we expect that an execution on a large cluster or supercomputer could allow the computation of 
transfer matrices of strips wider than what has been reached before.

\section*{Acknowledgment}
Special thanks to Pedro D. \'Alvarez for his explanations and useful advice on the computation of the limiting curves.
The authors would like to thank \textit{CONICYT} for sponsoring the PhD program of Crist\'obal A. Navarro, folio $N^o$ 21100750. 
This work was partially supported by the FONDECYT projects $N^o$ 1120495, $N^o$ 1120352 and the \textit{Millennium Nucleus Information and Coordination in Networks ICM/FIC P10-024F}. 

\bibliographystyle{elsarticle-num}
\bibliography{tmfamily}

\begin{thebibliography}{10}
\expandafter\ifx\csname url\endcsname\relax
  \def\url#1{\texttt{#1}}\fi
\expandafter\ifx\csname urlprefix\endcsname\relax\def\urlprefix{URL }\fi
\expandafter\ifx\csname href\endcsname\relax
  \def\href#1#2{#2} \def\path#1{#1}\fi

\bibitem{potts}
R.~B. Potts, Some generalized order-disorder transformation, in:
  Transformations, Proceedings of the Cambridge Philosophical Society, Vol.~48,
  1952, pp. 106--109.

\bibitem{PhysRevLett.43.799}
H.~W.~J. Bl\"ote, R.~H. Swendsen, First-order phase transitions and the
  three-state {P}otts model, Phys. Rev. Lett. 43 (1979) 799--802.

\bibitem{Chang_Shrock_2000}
S.-C. Chang, R.~Shrock, Exact {P}otts model partition functions on strips of
  the honeycomb lattice, Physica A: Statistical Mechanics and its Applications
  296~(1-2) (2000) 48.

\bibitem{Shrock_Tsai_1999}
R.~Shrock, S.-H. Tsai, Exact partition functions for {P}otts antiferromagnets
  on cyclic lattice strips, Physica A 275 (1999) 27.

\bibitem{Chang_Salas_Shrock_2002}
S.-C. Chang, J.~Salas, R.~Shrock, Exact {P}otts model partition functions on
  wider arbitrary-length strips of the square lattice, Journal of Statistical
  Physics 107~(5/6) (2002) 1207--1253.

\bibitem{Chang_Jacobsen_Salas_Shrock_2002}
S.-C. Chang, J.~L. Jacobsen, J.~Salas, R.~Shrock, Exact {P}otts model partition
  functions for strips of the triangular lattice, Physica A 286~(1-2) (2002)
  59.

\bibitem{ising_1925}
E.~Ising, Beitrag zur theorie des ferromagnetismus, Zeitschrift Für Physik
  31~(1) (1925) 253--258.

\bibitem{NYAS:NYAS627}
L.~Onsager, The effects of shape on the interaction of colloidal particles,
  Annals of the New York Academy of Sciences 51~(4) (1949) 627--659.

\bibitem{Woeginger:2003:EAN:885909.885927}
G.~J. Woeginger, Combinatorial optimization - eureka, you shrink!,
  Springer-Verlag New York, Inc., New York, NY, USA, 2003, Ch. Exact algorithms
  for NP-hard problems: a survey, pp. 185--207.

\bibitem{salas_sokal_1}
J.~Salas, A.~Sokal, Transfer matrices and partition-function zeros for
  antiferromagnetic {P}otts models. {I}. {G}eneral theory and square-lattice
  chromatic polynomial, Journal of Statistical Physics 104~(3-4) (2001)
  609--699.

\bibitem{1751-8121-43-31-315002}
J.~L. Jacobsen, Bulk, surface and corner free-energy series for the chromatic
  polynomial on the square and triangular lattices, Journal of Physics A:
  Mathematical and Theoretical 43~(31) (2010) 315002.

\bibitem{chang2009structure}
S.-C. Chang, R.~Shrock, Structure of the partition function and transfer
  matrices for the {P}otts model in a magnetic field on lattice strips, Journal
  of Statistical Physics 137~(4) (2009) 667--699.

\bibitem{salas_sokal_2011}
J.~Salas, A.~Sokal, Transfer matrices and partition-function zeros for
  antiferromagnetic {P}otts models {VI}. square lattice with extra-vertex
  boundary conditions, Journal of Statistical Physics 144~(5) (2011)
  1028--1122.

\bibitem{MGhaemi}
M.~Ghaemi, G.~A. Parsafar, Size reduction of the transfer matrix of
  two-dimensional {I}sing and {P}otts models, 2 4.

\bibitem{bedini}
A.~Bedini, J.~L. Jacobsen, A tree-decomposed transfer matrix for computing
  exact {P}otts model partition functions for arbitrary graphs, with
  applications to planar graph colourings, Journal of Physics A: Mathematical
  and Theoretical 43~(38) (2010) 385001.

\bibitem{blake2009survey}
G.~Blake, R.~G. Dreslinski, T.~Mudge, A survey of multicore processors, Signal
  Processing Magazine, IEEE 26~(6) (2009) 26--37.

\bibitem{duncan1990survey}
R.~Duncan, A survey of parallel computer architectures, Computer 23~(2) (1990)
  5--16.

\bibitem{navhitmat2014}
C.~A. Navarro, N.~Hitschfeld-Kahler, L.~Mateu, A survey on parallel computing
  and its applications in data-parallel problems using {GPU} architectures,
  Commun. Comput. Phys. 15 (2014) 285--329.

\bibitem{DBLP:conf/hpcc/NavarroHC13}
C.~A. Navarro, N.~Hitschfeld, F.~Canfora, Multi-core computation of transfer
  matrices for strip lattices in the potts model, in: 15th {IEEE} International
  Conference on High Performance Computing and Communications {\&} 2013 {IEEE}
  International Conference on Embedded and Ubiquitous Computing, {HPCC/EUC}
  2013, Zhangjiajie, China, November 13-15, 2013, 2013, pp. 125--134.

\bibitem{Wilf:2002:AC:560438}
H.~S. Wilf, Algorithms and Complexity, 2nd Edition, A. K. Peters, Ltd., Natick,
  MA, USA, 2002.

\bibitem{tutte}
W.~T. Tutte, A contribution to the theory of chromatic polynomials, J. Math 6
  (1954) 80--91.

\bibitem{potts_tutte_relation}
D.~Welsh, C.~Merino, The {P}otts model and the tutte polynomial, J. Math. Phys.
  43 (2000) 1127--1149.

\bibitem{sokal_2005}
A.~D. Sokal, The multivariate tutte polynomial (alias {P}otts model) for graphs
  and matroids, Surveys in Combinatorics 327 (2005) 173--226.

\bibitem{derrida1980transfer}
B.~Derrida, J.~Vannimenus, Transfer-matrix approach to percolation and
  phenomenological renormalization, Journal de Physique Lettres 41~(20) (1980)
  473--476.

\bibitem{baxter1982exactly}
R.~Baxter, Exactly solved models in statistical mechanics, Academic Press,
  1982.

\bibitem{Jacobsen2001701}
J.~Jacobsen, J.~Salas, Transfer matrices and partition-function zeros for
  antiferromagnetic {P}otts models. {II}. extended results for square-lattice
  chromatic polynomial, Journal of Statistical Physics 104~(3-4) (2001)
  701--723.

\bibitem{jacobsen2003}
J.~Jacobsen, J.~Salas, A.~Sokal, Transfer matrices and partition-function zeros
  for antiferromagnetic {P}otts models. {III}. triangular-lattice chromatic
  polynomial, Journal of Statistical Physics 112~(5-6) (2003) 921--1017.

\bibitem{jacobsen2006}
J.~Jacobsen, J.~Salas, Transfer matrices and partition-function zeros for
  antiferromagnetic {P}otts models : {IV}. chromatic polynomial with cyclic
  boundary conditions, Journal of Statistical Physics 122~(4) (2006) 705--760.

\bibitem{1179.82040}
J.~Salas, A.~D. Sokal, {Transfer matrices and partition-function zeros for
  antiferromagnetic {P}otts models. V. Further results for the square-lattice
  chromatic polynomial.}, J. Stat. Phys. 135~(2) (2009) 279--373.

\bibitem{jacobsen2007}
J.~Jacobsen, J.~Salas, Phase diagram of the chromatic polynomial on a torus,
  Nuclear Physics B 783~(3) (2007) 238--296.

\bibitem{Alvarez_Canfora_Reyes_Riquelme_2012}
P.~Alvarez, F.~Canfora, S.~Reyes, S.~Riquelme, {P}otts model on recursive
  lattices: some new exact results, The European Physical Journal B 85~(3)
  (2012) 1--13.

\bibitem{hartmann-2005-94}
A.~K. Hartmann, Partition function of two- and three-dimensional {P}otts
  ferromagnets for arbitrary values of q$>$0, Phys.rev.lett. 94 (2005) 050601.

\bibitem{Shrock_2000}
R.~Shrock, Exact {P}otts model partition functions on ladder graphs, Physica A:
  Statistical Mechanics and its Applications 283~(3-4) (2000) 73.

\bibitem{haggard_computing_tutte_polynomials}
G.~Haggard, D.~J. Pearce, G.~Royle, Computing tutte polynomials, ACM Trans.
  Math. Softw. 37 (2010) 24:1--24:17.

\bibitem{DBLP:conf/focs/BjorklundHKK08}
A.~Bj{\"{o}}rklund, T.~Husfeldt, P.~Kaski, M.~Koivisto, Computing the tutte
  polynomial in vertex-exponential time, in: 49th Annual {IEEE} Symposium on
  Foundations of Computer Science, {FOCS} 2008, October 25-28, 2008,
  Philadelphia, PA, {USA}, 2008, pp. 677--686.

\bibitem{halverson2005partition}
T.~Halverson, A.~Ram, Partition algebras, European Journal of Combinatorics
  26~(6) (2005) 869--921.

\bibitem{Dutton:1986:CEB:10987.10992}
R.~D. Dutton, R.~C. Brigham, Computationally efficient bounds for the catalan
  numbers, Eur. J. Comb. 7~(3) (1986) 211--213.

\bibitem{salas2002}
S.-C. Chang, J.~Salas, R.~Shrock, Exact {P}otts model partition functions for
  strips of the square lattice, Journal of Statistical Physics 107~(5-6) (2002)
  1207--1253.

\bibitem{coxeter1973regular}
H.~S.~M. Coxeter, Regular polytopes, Courier Dover Publications, 1973.

\bibitem{Henk:1997:BPC:285869.285884}
M.~Henk, J.~Richter-Gebert, G.~M. Ziegler, Handbook of discrete and
  computational geometry, CRC Press, Inc., Boca Raton, FL, USA, 1997, Ch. Basic
  properties of convex polytopes, pp. 243--270.

\bibitem{fo_95}
I.~Foster, Designing and building parallel programs: Concepts and tools for
  parallel software engineering, Addison-Wesley Longman Publishing Co., Inc.,
  Boston, MA, USA, 1995.

\bibitem{Chapman:2007:UOP:1370966}
B.~Chapman, G.~Jost, R.~V.~D. Pas, Using OpenMP: Portable Shared Memory
  Parallel Programming (Scientific and Engineering Computation), The MIT Press,
  2007.

\bibitem{mpi}
M.~P. Forum, Mpi: A message-passing interface standard, Tech. rep., Knoxville,
  TN, USA (1994).

\bibitem{Bauer20021}
C.~Bauer, A.~Frink, R.~Kreckel, Introduction to the ginac framework for
  symbolic computation within the c++ programming language, Journal of Symbolic
  Computation 33~(1) (2002) 1 -- 12.

\bibitem{1751-8121-47-13-135001}
J.~L. Jacobsen,
  \href{http://stacks.iop.org/1751-8121/47/i=13/a=135001}{High-precision
  percolation thresholds and {P}otts-model critical manifolds from graph
  polynomials}, Journal of Physics A: Mathematical and Theoretical 47~(13)
  (2014) 135001.
\newline\urlprefix\url{http://stacks.iop.org/1751-8121/47/i=13/a=135001}

\bibitem{DBLP:journals/tc/EagerZL89}
D.~L. Eager, J.~Zahorjan, E.~D. Lazowska, Speedup versus efficiency in parallel
  systems, IEEE Trans. Computers 38~(3) (1989) 408--423.

\bibitem{wilkinson1999parallel}
B.~Wilkinson, C.~M. Allen, Parallel programming, page 7, Prentice hall New
  Jersey, 1999.

\bibitem{gustafson1990fixed}
J.~L. Gustafson, Fixed time, tiered memory, and superlinear speedup, in:
  Proceedings of the Fifth Distributed Memory Computing Conference (DMCC5),
  1990, pp. 1255--1260.

\end{thebibliography}
\end{document}